\newif\ifshort
\shorttrue
\newif\ifappendix
\appendixtrue

\documentclass[11pt]{article}

\usepackage[a4paper,includefoot,margin=3cm]{geometry}
\usepackage{amsmath, amssymb, amsthm}
\usepackage{nicefrac}
\usepackage[mathic]{mathtools}
\usepackage{thm-restate}
\usepackage{graphicx,scalerel}
\usepackage[dvipsnames]{xcolor}
\usepackage{subcaption}
\usepackage{multirow}
\usepackage{authblk}

\usepackage{bm}
\usepackage{contour}
\usepackage{framed}

\usepackage[textsize=footnotesize,backgroundcolor=green!40!white,linecolor=black!60,obeyFinal]{todonotes}
\usepackage[T1]{fontenc}
\usepackage[utf8]{inputenc}

\usepackage{dsfont, paralist, microtype, enumerate}
\usepackage{booktabs, tabularx}

\usepackage[pdfdisplaydoctitle,menucolor=orange!40!black,filecolor=magenta!40!black,urlcolor=blue!40!black,linkcolor=red!40!black,citecolor=green!40!black,ocgcolorlinks]{hyperref} %

\newcommand{\ExternalLink}{%
    \tikz[color=magenta, x=1.2ex, y=1.2ex, baseline=-0.05ex]{%
        \begin{scope}[x=1ex, y=1ex]
            \clip (-0.1,-0.1) 
                --++ (-0, 1.2) 
                --++ (0.6, 0) 
                --++ (0, -0.6) 
                --++ (0.6, 0) 
                --++ (0, -1);
            \path[draw, 
                line width = 0.5, 
                rounded corners=0.5] 
                (0,0) rectangle (1,1);
        \end{scope}
        \path[draw, line width = 0.5] (0.5, 0.5) 
            -- (1, 1);
        \path[draw, line width = 0.5] (0.6, 1) 
            -- (1, 1) -- (1, 0.6);
        }
}

\DeclareUnicodeCharacter{1EF3}{\`y}
\usepackage[
	backend=biber,
	isbn=false,
	url=false,
	doi=true,
	backref=true,
	hyperref=true,
	natbib=true,
	maxcitenames=3,
	maxbibnames=99,
	sortcites]{biblatex}
\renewbibmacro{in:}{\ifentrytype{article}{}{\printtext{\bibstring{in}\intitlepunct}}}
\newbibmacro{doilink}[1]{%
	\iffieldundef{doi}{%
		\iffieldundef{url}{%
			#1%
        }{%
		\href{\thefield{url}}{$\ExternalLink$}%
		}%
	}{%
		\href{https://doi.org/\thefield{doi}}{$\ExternalLink$}%
	}%
}
\DeclareFieldFormat*{doi}{\usebibmacro{doilink}{#1}}

\bibliography{strings-long,references}

\usepackage[capitalise]{cleveref}
\crefname{alg}{Algorithm}{Algorithms}
\Crefname{alg}{Algorithm}{Algorithms}
\crefname{axiom}{}{}
\creflabelformat{axiom}{(#2I#1#3)}

\usepackage{url}
\usepackage{paralist}

\usepackage{tikz}
\usetikzlibrary{angles, arrows, arrows.meta, calc, decorations.pathmorphing, decorations.pathreplacing, fit, intersections, matrix, calligraphy, backgrounds, math, shapes, shapes.geometric, spy, positioning, quotes}

\tikzstyle{path} = [color=black!10,line cap=round, line join=round, line width=12pt]
\tikzset{
	vertex/.style={
		circle,
		draw=black,
		solid,
		thick,
		fill=white,
		minimum size=2mm,
		inner sep=0,
	},
	edge/.style={
		draw=black,
		semithick,
	},
	arc/.style={
		edge,
		->,
		>=stealth',
	},
	colornode/.style = {
		circle,
		draw=#1!70!black,
		very thick,
		fill=#1
	},
	smallcolornode/.style = {
		colornode=#1,
		thick,
		scale=0.65
	}
}

\makeatletter

\def\@maketitle{%
  \newpage
  \null
  \vskip 2em%
  \begin{center}%
  \let \footnote \thanks
    {\Large\bf \@title \par}%
    \vskip 1.5em%
    {\large
      \lineskip .5em%
      \begin{tabular}[t]{c}%
        \@author
      \end{tabular}\par}%
    \vskip 1em%
    {\large \@date}%
  \end{center}%
  \par
  \vskip 1.5em}

\def\url@leostyle{%
  \@ifundefined{selectfont}{\def\UrlFont{\sf}}{\def\UrlFont{\small\ttfamily}}}
\makeatother
\usepackage{thmtools}
\theoremstyle{plain}
\newtheorem{theorem}{Theorem}
\crefname{theorem}{Theorem}{Theorems}
\Crefname{theorem}{Thm{.}}{Thms{.}}
\newtheorem{lemma}[theorem]{Lemma}

\newtheorem{corollary}[theorem]{Corollary}
\newtheorem{observation}[theorem]{Observation}
\Crefname{observation}{Obs{.}}{Obs{.}}
\crefname{observation}{Observation}{Observations}

\theoremstyle{definition}

\newtheorem{fact}[theorem]{Fact}
\newtheorem{algo}{Algorithm}

\theoremstyle{remark}

\declaretheorem[style=definition,name=Construction,qed=$\diamond$]{construction}

\newcommand{\prob}[1]{\textnormal{\textsc{#1}}}

\newcommand{\probDefDecision}[3]{
	\begin{center}
	\begin{minipage}{\columnwidth-6em}
		\noindent
		\prob{#1}
		\vspace{5pt}\\
		\setlength{\tabcolsep}{3pt}
		\begin{tabularx}{\textwidth}{@{}lX@{}}
			\textbf{Input:}     & #2 \\
			\textbf{Question:}  & #3
		\end{tabularx}
	\end{minipage}
	\end{center}
}

\newcommand{\probDefOptimization}[4]{
	\begin{center}
	\begin{minipage}{\columnwidth-6em}
		\noindent
		\prob{#1}
		\vspace{5pt}\\
		\setlength{\tabcolsep}{3pt}
		\begin{tabularx}{\textwidth}{@{}lX@{}}
			\textbf{Input:}     & #2 \\
			\textbf{Solution:}  & #3 \\
			\textbf{Objective:} & #4
		\end{tabularx}
	\end{minipage}
	\end{center}
}

\DeclarePairedDelimiterX{\abs}[1]{\lvert}{\rvert}{#1}
\DeclarePairedDelimiterX{\norm}[1]{\lVert}{\rVert}{#1}
\DeclarePairedDelimiterX{\ceil}[1]{\lceil}{\rceil}{#1}

\newcommand{\NN}{\ensuremath{\mathds{N}}}

\newcommand{\RR}{\ensuremath{\mathds{R}}}

\newcommand{\III}{\ensuremath{\mathcal{I}}}

\newcommand{\NP}{\ensuremath{\mathrm{NP}}}
\newcommand{\Poly}{\ensuremath{\mathrm{P}}}

\newcommand{\bigO}{\ensuremath{\mathcal{O}}}

\newcommand{\eps}{\varepsilon}

\newcommand{\ceq}{\ensuremath{\coloneqq}}

\usepackage{xargs}
\newcommandx{\set}[2][1=1]{\ensuremath{\{#1,\ldots,#2\}}} %
\newcommandx{\tlog}[3][1=,3=]{\log_{#1}^{#3}(#2)}
\newcommandx{\ith}[2][1=th]{#2\nobreakdash-#1}

\newcommand{\ttup}[1]{\textup{\texttt{#1}}} %
\newcommand{\scup}[1]{\textup{\textsc{#1}}} %

\DeclareMathOperator{\OPT}{OPT}

\DeclareMathOperator{\PP}{\mathrm{Pr}} %
\DeclareMathOperator{\EE}{\mathbb{E}} %

\DeclareMathOperator{\Int}{\mathrm{Int}} %

\newcommand{\SU}{{\mathcal{S}U}}

\usepackage{etoolbox}
\newcommand{\appsymb}{$\bigstar$}

\ifshort{}\ifappendix{}
\newcommand{\appref}[1]{\hyperref[proof:#1]{\appsymb}}
\else{}
\newcommand{\appref}[1]{\appsymb}
\fi{}
\else{}
\newcommand{\appref}[1]{}
\fi{}

\title{Approximability of Electrical Distribution Network Reconfiguration for General Graphs}

\author[1,2]{Christian~Wallisch}

\author[1,3,4]{Andrea~Benigni}

\author[1,3]{Carsten~Hartmann}

\author[5]{Leon~Kellerhals}

\affil[1]{\normalsize Energy Systems Engineering (ICE-1), Institute of Climate and Energy Systems, Forschungszentrum J\"{u}lich, Germany}
\affil[2]{\normalsize Hasso Plattner Institute, University of Potsdam, Germany}
\affil[3]{\normalsize RWTH Aachen University, Germany}
\affil[4]{\normalsize JARA-Energy, J\"{u}lich, Germany}
\affil[5]{\normalsize TU Clausthal, Germany}
\date{}

\begin{document}

\maketitle

\begin{abstract}
  Electrical distribution networks are regional, medium- and low-voltage power grids connecting energy sources to individual households and businesses with given power demands.
  While these networks contain redundant power lines for reliability, they are typically operated in a radial (spanning tree) configuration by opening and closing switches on the lines.
  The challenge is to find a spanning tree that
  minimizes the sum of the resistive power losses: The power loss of a line $e$ is its resistance $r(e)$ times the squared current $f(e)^2$ flowing across the line.

  We study approximation algorithms for this problem, known as \textsc{Distribution Network Reconfiguration (DNR)}.
  We give an $n$-approximation algorithm and, via a new NP-hardness for planar \textsc{Balanced Connected Partition} with a fixed number of parts, show that no $n^{1-\varepsilon}$-approximation is possible even on planar graphs unless $\mathrm{P}=\mathrm{NP}$, for any $\varepsilon>0$.
  Since the approximation hardness holds only if there are many sources, we focus on \textsc{$k$-DNR} with $k$ sources; this is motivated by traditional distribution networks, where oftentimes $k = 1$.
  For $2$-DNR, we give an approximation lower bound of $\Omega(\log^2 n)$ conditioned on $\mathrm{P}\neq\mathrm{NP}$.
  For $1$-DNR, which is equivalent to finding an uncapacitated confluent flow minimizing the squared Euclidean norm, we prove APX-hardness and give an $\mathcal{O}(\sqrt{n})$-approximation for uniform line resistances, answering an open question by \citet{gupta2022electrical}.%
\end{abstract}

\section{Introduction}

When electric current flows through a power line, some power is converted into heat and thereby lost.
This is captured by Joule's law: The power loss on a line is the product of the line's resistance and the square of the current flowing through it.
This means that heavily loaded lines incur disproportionately large losses.
It is therefore essential to avoid heavy loads on lines when distributing energy to households and businesses.
While distribution networks (interpreted as undirected graphs whose edges are power lines) offer many connections to the consumers not least to offer redundancy \cite{guo2023review},
they are typically operated in a radial configuration: among all available lines, the active ones must form a spanning tree.
For any such radial configuration, the production (negative demands) and consumption (positive demands) at the network vertices uniquely determine the flow on each active line, and hence the total power loss, which should be minimized.
Formally, this is captured by the following optimization problem.
\probDefOptimization{Distribution Network Reconfiguration (DNR)}
{A connected undirected graph $G = (V,E)$, a demand function $d \colon V \to \RR$ satisfying $\sum_{x \in V} d(x) = 0$, and a resistance function $r \colon E \to \RR_{\geq 0}$.}
{A spanning tree $T$ of $G$.}
{
  Minimize
  $\displaystyle L(T) \coloneqq \sum_{\{x,y\} \in E(T)} r(\{x, y\}) \left(\sum_{v \in V(T_{x \setminus y})} d(v)\right)^2$.
}
Herein, for an edge $\{x, y\}$, the trees $T_{x \setminus y}$ and $T_{y \setminus x}$ are the unique components obtained by removing $\{x, y\}$ from $T$, with $x \in V(T_{x \setminus y})$ and $y \in V(T_{y \setminus x})$.
Notably, it does not matter whether we pick $T_{x \setminus y}$ or $T_{y \setminus x}$ in $L_T$:
as $G$ is connected and the sum of the demands is $0$, we have $\sum_{v \in V(T_{x \setminus y})} d(v) = - \sum_{v \in V(T_{y \setminus x})} d(v)$.

The necessity of finding an energy-efficient radial configuration is evident.
A recent CEER report \cite{ceer2025power} illustrates the scale: across 33 European countries, reported distribution losses ranged from 1.95\% to 22.63\%, with a median of 6.31\%.\footnote{Because of accounting difficulties, these values should be understood as upper bounds.}
Moreover, temporal variability is becoming more pronounced as low-carbon technologies such as solar panels, electric vehicles, and heat pumps increase network utilization and make demand and generation patterns more irregular \cite{damianakis2025grid};
therefore,
modern distribution networks have remote-controlled switches, enabling frequent reconfiguration in response to load variations on short time scales \cite{behbahani2024comprehensive}.
In this temporally demanding setting, the challenge now is to quickly find an efficient reconfiguration.

A recent review \cite{behbahani2024comprehensive} reports that hundreds of papers on distribution network reconfiguration appear each year, with loss minimization as the most studied objective and most papers contributing heuristic approaches.
Curiously, while the first work on \scup{DNR} dates back to \citeyear{merlin1975search}~\cite{merlin1975search},
its \NP-hardness was only established in \citeyear{khodabakhsh2018submodular}~\cite{khodabakhsh2018submodular}.
More recently, \citet{gupta2022electrical,ito2026loss} devised the first (polynomial-time) approximation algorithms for the special case \scup{1-DNR}, where we have a single power source (i.e., one vertex with negative demand); among them an $\bigO(n)$-approximation for general graphs and a $\bigO(\sqrt{n})$-approximation for grid graphs with uniform resistances \cite{gupta2022electrical}.
Our work continues the study of the approximability for \scup{DNR}.

\subsection{Our contributions}

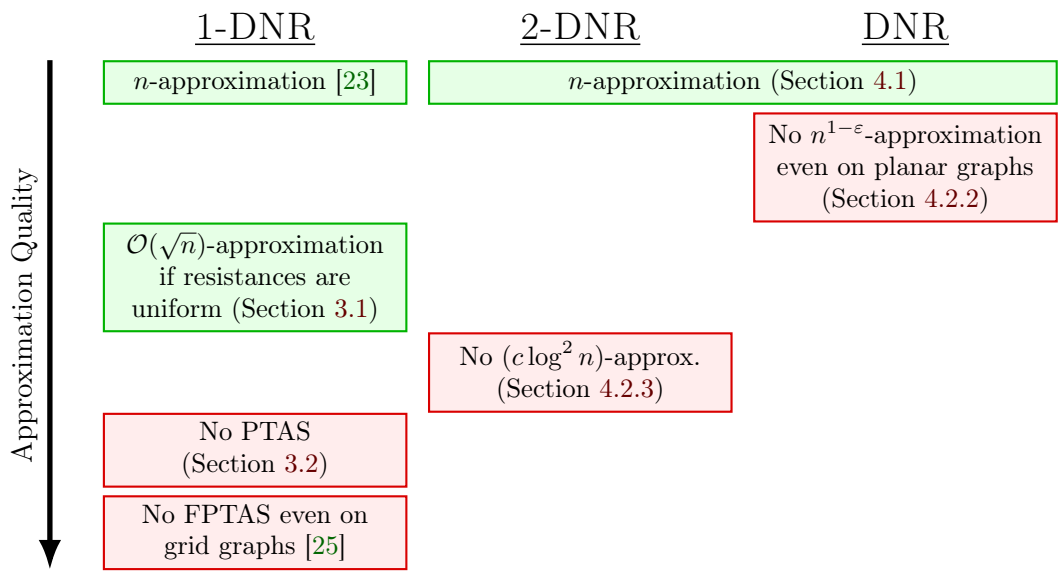
\begin{figure}[t]
    \centering
    
    \begin{tikzpicture}[
            yscale=-1,
            good/.style={
                draw=green!70!black,
                fill=green!10,
                thick,
                rectangle,
                align=center,
                inner sep=3pt,
                minimum width=4cm,
                font={\small},
            },
            bad/.style={
                draw=red!85!black,
                fill=red!7,
                thick,
                rectangle,
                align=center,
                inner sep=3pt,
                minimum width=4cm,
                font={\small},
            },
        ]

        \tikzmath{
            \w = 4;
            \xsep = 0.3;
            \ysep = 0.1;
            \titlesep = 0.7;
            \arrowsep = 0.7;
            \x1 = 0.5*\w;
            \x2 = 1.5*\w + \xsep;
            \xk = 2.5*\w + 2*\xsep;
            \doublew = 2*\w+\xsep;
        }

        \node (s1) at (\x1,0) {\Large \underline{\scup{1-DNR}}};
        \node (s2) at (\x2,0) {\Large \underline{\scup{2-DNR}}};
        \node (sk) at (\xk,0) {\Large \underline{\scup{DNR}}};

        \node[good] (first1) at (\x1,\titlesep) {$n$-approximation \cite{gupta2022electrical}};

        \node[good,minimum width=\doublew cm] (first2) at (2*\w+1.5*\xsep,\titlesep) {$n$-approximation (\cref{sub:multi-approx})};

        \coordinate (firstLow) at (\xk,|-first2.south);

        \node[bad,below] (second) at ([yshift=\ysep cm]firstLow) {No $n^{1-\varepsilon}$-approximation\\ even on planar graphs\\ (\cref{sub:multi-hardness})};

        \coordinate (secondLow) at (\x1,|-second.south);

        \node[good,below] (third) at (secondLow) {$\bigO(\sqrt{n})$-approximation\\ if resistances are \\ uniform (\cref{sub:single-approx})};

        \coordinate (thirdLow) at (\x2,|-third.south);

        \node[bad,below] (fourth) at (thirdLow) {No $(c \log^2 n)$-approx.\\ (\cref{sub:two-hardness})};

        \coordinate (fourthLow) at (\x1,|-fourth.south);

        \node[bad,below] (fifth) at (fourthLow) {No PTAS\\ (\cref{sub:single-hardness})};

        \coordinate (fifthLow) at (\x1,|-fifth.south);

        \node[bad,below] (sixth) at ([yshift=\ysep cm]fifthLow) {No FPTAS even on\\ grid graphs \cite{ito2024loss}};

        \draw [-{Latex[length=4mm]},line width=2pt] ([xshift=-\arrowsep cm]first1.north west) -- node[above,sloped] {Approximation Quality} ([xshift=-\arrowsep cm]sixth.south west);

    \end{tikzpicture}

    \caption[Illustration of the Steps Used to Prove NP-Hardness of $c \log^2 N$-Approximation]{%
        Approximation Landscape.
        The three columns show results for the single-source, the two-source, and the general \scup{DNR} problem.
    }
    \label{fig:overview}
\end{figure}

Our contributions can be divided into two parts; an overview is given in \cref{fig:overview}.

First, we continue the study of approximability of \scup{1-DNR} (\cref{sec:single}).
As our main result, we show that \scup{1-DNR} admits a polynomial-time $\bigO(\sqrt{n})$-approximation algorithm on general graphs with uniform resistances; thereby resolving an open question by \citet{gupta2022electrical}.
Our result makes use of an equivalent (known) flow formulation of \scup{DNR} and the fact that \scup{1-DNR} is equivalent to a \emph{confluent flow} problem.
This paves the way for a randomized rounding-based approach.
As soon as we allow non-uniform resistances, we can show \scup{1-DNR} to not admit a PTAS unless $\Poly = \NP$; this result holds even for binary resistances ($r(e) \in \{0,1\}$ for all $e \in E$) and polynomial demands.

Second, we extend the approximability analysis to \scup{DNR} with more than one source (\cref{sec:multi}).
This perspective is especially relevant today, as distributed energy resources that support decarbonization , such as solar panels and battery storage, are integrated into distribution networks \cite{navidi2023coordinating}.
Here, we show that any acyclic\footnote{By acyclic, we mean free of \emph{undirected} cycles.} flow routing the demand is an $n$-approximation for \scup{DNR}.
Moreover, we show that, even on planar graphs, a polynomial-time $n^{1-\eps}$-approximation algorithm for \scup{DNR} implies $\Poly=\NP$ for any $\eps > 0$, even for binary resistances and polynomial demands.
To obtain the hardness, we establish the \NP-hardness of a balanced graph partitioning problem on planar graphs which was conjectured by \citet{dyer1985complexity} in~\citeyear{dyer1985complexity}.
Finally, we provide an $\Omega(\log^2 n)$ approximation lower bound for \scup{2-DNR}, again conditioned on $\Poly \ne \NP$ and even for binary resistances and polynomial demands.

\subsection{Related work}

\paragraph*{Network reconfiguration.}
Motivated by the 1973 oil crisis, \textcite{merlin1975search} initiated the study of \scup{DNR}.
Early work established heuristics \cite{civanlar1988distribution,baran1989network,merlin1975search,shirmohammadi1989reconfiguration},
which was later complemented by various metaheuristic approaches \cite{nara1992genetic,chang1994network,su2005distribution,rao2013power,nguyen2015distribution},
as well as approaches based on mathematical programming \cite{aoki1987normal}, machine learning \cite{kim1993artificial}, and quantum computing \cite{silva2023quantum}.

Despite the sustained interest from the engineering and heuristics communities, theoretical work on \scup{DNR} remains comparatively sparse.
The first explicit hardness result was given by \textcite{khodabakhsh2018submodular}, showing that \scup{1-DNR} is strongly \NP-hard even for uniform demands and resistances, which rules out an FPTAS unless $\mathrm{P}=\mathrm{NP}$.
\textcite{gupta2022electrical}, in addition to their above-mentioned results, give a $2$-approximation algorithm on grid graphs with uniform demands and resistances.
\textcite{ito2026loss} improve the asymptotic factor from $2$ to $9/8$.
Moreover, they show that \scup{1-DNR} is weakly \NP-hard on grid graphs of height three, and strongly \NP-hard on arbitrary grid graphs.

Further results with guaranteed error bounds are given by
\textcite{inoue2014distribution}, who present an exponential-time algorithm for \scup{DNR},
and by
\textcite{khodabakhsh2018submodular}, who give an algorithm for \scup{1-DNR} whose guarantee is stated relative to an upper bound on the objective value rather than to the optimum.

\paragraph*{Transshipment and confluent flows.}
As we will show in \cref{sec:prelim}, \scup{DNR} is equivalent to the problem of finding an acyclic flow $f$ routing the demands so as to minimize $\sum_{\{x,y\}\in E} r(\{x, y\}) f(x,y)^2$ (assuming skew symmetry, i.e., $f(x, y) = -f(y, x)$).
\scup{DNR} therefore has two features that distinguish it from the \scup{Transshipment} problem,\footnote{Also known as \scup{Uncapacitated Min-Cost Flow}; see \cref{sub:multi-approx} for a definition.} which is solvable in almost-linear time \cite{chen2025maximum,brand2023deterministic}.
These features are the quadratic (instead of linear) costs and the acyclicity constraint.
Notably, one can evade the \NP-hardness by dropping one of the two distinguishing features:
Dropping the acyclicity constraint (i.e., allowing flow on all lines) yields the electrical flow problem, which can be solved (exactly) through convex quadratic programming~\cite{kozlov1980polynomial} and can be approximated to arbitrary precision in almost linear time~\cite{christiano2011electrical}.
Replacing the quadratic with linear costs is a well-known polynomial-time solvable variant of \scup{Transshipment}~\cite[Chapter 11.2]{ahuja1993network}.

A flow is \emph{confluent} if every vertex sends out flow over at most one edge.
Earlier work studied the problem of finding a confluent flow minimizing the congestion, i.e., the maximum of the flows over edges.
In particular, the foundational paper by \textcite{chen2006meet} gives a $\widetilde{\bigO}(\sqrt{n})$-approximation algorithm based on randomized rounding for uniform supplies and capacities, where capacities are the inverse of resistances in our model.\footnote{In electrical networks, the inverse of resistance is known as conductance.}
For arbitrary supplies and uniform capacities, \textcite{chen2007almost} essentially settle the approximability of minimum-congestion confluent flow: they present a $(1 + \ln k)$-approximation algorithm and prove NP-hardness within $\log_2 k$, where $k$ is the number of sink-adjacent vertices.
\scup{1-DNR} is equivalent to finding a confluent flow with a quadratic objective, which has not been studied previously to the best of our knowledge.

\subsection{An Engineer's Perspective}

This work is a study of \scup{DNR} in its most elementary form, matching the formalizations used by the above-mentioned theoretical works~\cite{gupta2022electrical,ito2024loss}.
We highlight three aspects in which this model differs from its real-world counterpart.
First, we replace a full load-flow model with a linearized, resistance-only surrogate: reactive power is neglected and voltages are treated as fixed at nominal, so that the flow across a line is determined solely by the downstream demand it serves.
Second, we account for losses in the objective but not in the flow-conservation constraints; the additional injection needed to compensate for a resistive loss, which would itself flow upstream and incur further losses, is a second-order term that we omit.
Third, we assume that every line is switchable, as is common in formulations of radiality-constrained DNR problems \cite{Lavorato2012,Wang2020}, whereas in practice only a fraction of the lines can be switched \cite{Celli2004}, so that only a subset of the spanning trees is realizable.
These simplifications mean that the inapproximability (but not the approximability) results presented in this work carry over to the real-world setting.

\section{Preliminaries}
\label{sec:prelim}

Let $G = (V, E)$ be an undirected, connected graph
with an electrical \emph{resistance} $r(e) \in \RR_{\ge 0}$ for each edge $e \in E$
and demands $d \colon V \to \RR$ satisfying $\sum_{x \in V} d(x)=0$.
Let $A \ceq \{(x,y) \in V \times V \mid \{x,y\} \in E\}$ be the set of oriented edges of $G$.
A \emph{flow} satisfying $d$ is an assignment $f \colon A \to \RR$ such that
\begin{align*}
    &\sum_{x \in N_G(y)} f(x,y) = d(y) && \text{for all $y \in V$ \emph{(flow conservation)}, and} \\
    &f(x,y) = -f(y,x) && \text{for all $(x,y) \in A$ \emph{(skew symmetry)}.}
\end{align*}
The \emph{support} $G_f = (V,E_f)$ of a flow $f$ is the subgraph of $G$ containing all edges with nonzero flow, i.e., $E_f \ceq \{\{x,y\} \in E \mid f(x,y) \neq 0\}$.
A flow $f$ is \emph{acyclic} if its support $G_f$ is acyclic.

The \emph{power loss} of a flow $f$ is
\begin{displaymath}
    L(f) \ceq \sum_{\{x,y\} \in E} r(\{x,y\}) \cdot \big( f(x,y) \big)^2.
\end{displaymath}
Note that the power loss is well-defined because $f(x,y)^2 = f(y,x)^2$ by skew symmetry.

Observe that a spanning tree $T$ of $G$ uniquely determines a flow $f_T$,
with
\[
    f_T(x,y) \ceq \sum_{v \in V(T_{y \setminus x})} d(v) = -\sum_{v \in V(T_{x \setminus y})} d(v)
\]
for every $\{x, y\} \in E(T)$, and $f_T(x, y) \coloneqq 0$ for every $\{x, y\} \notin E(T)$.
Note that $f_T$ satisfies the demands and its support is a subgraph of $T$, and $L(T) = L(f_T)$.

Conversely, the support $G_f$ of any acyclic flow $f$ can be extended to a spanning tree by adding edges to connect the components of $G_f$.

\section{Single-Source Networks}\label{sec:single}

\subsection{A Single-Source Approximation Algorithm}\label{sub:single-approx}

\textcite{gupta2022electrical} provide an $\bigO(\sqrt{n})$-approximation algorithm for \scup{1-DNR} on grid graphs with uniform resistances (and the source located in a corner).
They discuss the challenge of generalizing their result to planar graphs.
In this section, we provide an $\bigO(\sqrt{n})$-approximation algorithm that even works on general graphs.
We achieve this by a largely independent approach based on randomized rounding.

\begin{theorem}\label{thm:1dnr-tractability}
    There is a randomized $\bigO(\sqrt{n})$-approximation algorithm for \scup{1-DNR} with uniform resistances.
\end{theorem}

We will interpret \scup{1-DNR} as a confluent flow problem.
We say that a flow $f$ is \emph{confluent} if every vertex sends out flow over at most one edge, that is, for every $x \in V$, there is at most one $y \in N_G(x)$ with $f(x,y) > 0$.
In the literature on confluent flows, the standard assumption is that the network contains a single sink instead of a single source.
These settings are equivalent by multiplication of all demands by $-1$.
To match the terminology of the literature, we assume to have a single sink.

Next, we describe the close relationship between solutions to the \scup{1-DNR} problem and confluent flows.
In a spanning tree, each source has a unique path to the sink.
Hence, if a distribution network is in a spanning tree configuration, then each source sends out electrical current over exactly one edge, which means that the flow is confluent.

Conversely, each confluent flow becomes an acyclic flow by removing directed flow cycles; see the textbook by \textcite[Chapter 3.5]{ahuja1993network} for details on directed flow cycle removal.
Note that a flow cycle with opposed flow directions cannot occur in a confluent flow, since otherwise some vertex on the cycle would send out flow over two edges.
Moreover, directed flow cycle removal can be performed in polynomial time and only decreases the flow over the edges it modifies, which implies that the power loss does not increase.

When resistances are uniform ($r \equiv 1$), then \scup{1-DNR} minimizes the squared Euclidean norm of the flow on the edges.
Hence, under the above considerations, \scup{1-DNR} with uniform resistances is the problem of finding an optimal confluent flow under the squared Euclidean norm.

In contrast to confluent flows, we refer to general flows, which may split at vertices, as \emph{splittable}.
As noted by \textcite{gupta2022electrical}, the gap between the minimum power loss of a splittable flow and that of a confluent flow can be as large as $\Omega(n)$.
This presents a significant obstacle to the use of splittable-flow relaxations.
The same difficulty arises when minimizing the congestion, i.e., the $\ell_\infty$-norm \cite{lorenz2001how,chen2006meet}.
Addressing this, \textcite{chen2006meet} make a crucial observation regarding confluent flows: since the entire outflow of each vertex is carried by a single edge, the vertex outflows directly correspond to the edge flows.
This enables us to adopt a helpful vertex-centric view.

We remark that \textcite{chen2006meet} present a randomized rounding algorithm for the $\ell_\infty$-norm that closely resembles our algorithm.
The analyses, however, differ fundamentally: they bound higher moments of the rounded flow through each vertex, whereas we use a coupling argument.
It is not clear how to extend their analysis to the squared Euclidean norm.
Moreover, they assume uniform demands at the sources and obtain an approximation ratio that is larger by a polylogarithmic factor.

We proceed with the description of our algorithm.
Let $\III = (G = (V,E), d, r)$ be an instance of \scup{1-DNR} with $r \equiv 1$.
Let $t \in V$ be the single sink.
Given a flow $f$, the \emph{outflow} of a vertex $x \in V$ is
\begin{displaymath}
    f^\mathrm{out}(x) \ceq \sum_{y \in N_G(x)} \max(0, f(x,y)).
\end{displaymath}
If $f$ is a confluent flow, then the vertex outflows determine the loss:
\begin{displaymath}
    L(f)
    = \sum_{x \in V} \big( f^\mathrm{out}(x) \big)^2.
\end{displaymath}

Let $A \ceq \{(x,y) \in V \times V \mid \{x,y\} \in E\}$ be the set of oriented edges of $G$.
We consider the following splittable-flow relaxation:
\begin{align*}
    \text{minimize }& \sum_{x \in V} \big( f^\mathrm{out}(x) \big)^2 \\
    \text{subject to }& \sum_{x \in N_G(y)} f(x,y) = d(y) && \text{for all $y \in V$ (flow conservation), and} \\
    &f(x,y) = -f(y,x) && \text{for all $(x,y) \in A$ (skew symmetry),}
\end{align*}
where $f(x,y)$ is a variable for each $(x,y) \in A$, and
\begin{displaymath}
     f^\mathrm{out}(x) \ceq \sum_{y \in N_G(x)} \max(0, f(x,y)).
\end{displaymath}
Using a standard transformation, we obtain an equivalent quadratic program (QP).
More precisely, we make $f^\mathrm{out}(x)$ a variable for each $x \in V$, introduce a variable $f^\mathrm{max}(x,y)$ for each $(x,y) \in A$, and add the constraints
\begin{align*}
    &f^\mathrm{max}(x,y) \geq f(x,y) && \text{for all $(x,y) \in A$,} \\
    &f^\mathrm{max}(x,y) \geq 0 && \text{for all $(x,y) \in A$, and} \\
    &f^\mathrm{out}(x) \geq \sum_{y \in N_G(x)} f^\mathrm{max}(x,y) && \text{for all $x \in V$.}
\end{align*}

Let $f_\mathrm{split}$ be an optimal splittable flow solving the above QP relaxation.
Since the QP is convex, we can compute $f_\mathrm{split}$ in polynomial time \cite{kozlov1980polynomial}.
Moreover, $f_\mathrm{split}$ contains no directed flow cycles because of its optimality.

Next, we describe the randomized rounding procedure for obtaining a confluent flow from $f_\mathrm{split}$.
Let $S' \ceq \{x \in V \mid f_\mathrm{split}^\mathrm{out}(x) > 0\}$ be the set of vertices with positive outflow.
We define the graph $G' \ceq (V', A')$ where $V' \ceq S' \cup \{t\}$ and $A' \ceq \{(x,y) \in A \mid f_\mathrm{split}(x,y) > 0\}$.
Since $f_\mathrm{split}$ contains no directed flow cycles, $G'$ is a directed acyclic graph (DAG).
Moreover, each vertex of $S'$ has at least one outgoing arc in $G'$ because of flow conservation.
Note that every vertex that is contained in $G$ but not in $G'$ has an outflow of zero and hence zero demand.

To obtain a confluent flow, we select one outgoing arc for each vertex $x \in S'$: we choose the arc $(x,y) \in A'$ with probability
\begin{displaymath}
    p_x(y) \ceq \frac{f_\mathrm{split}(x,y)}{f_\mathrm{split}^\mathrm{out}(x)}.
\end{displaymath}
Note that $p_x$ is indeed a probability distribution.

Since $G'$ is a DAG, no cycles are formed: the chosen arcs form a spanning arborescence (i.e., a rooted tree in which every directed edge points to the root) of $G'$ directed toward the sink $t$.
This spanning arborescence induces a unique confluent flow $f_\mathrm{conf}$ over $G'$.
As $G$ and $G'$ differ only by demand-$0$ vertices, $f_\mathrm{conf}$ is also a confluent flow over $G$.
(Recall that every vertex must send out flow over \emph{at most} one edge.)
We return $f_\mathrm{conf}$ as the result.

\begin{framed}
    \begin{algo}[Summary]\label{alg:single}
        \emph{Input:} An instance $\III = (G=(V,E), d, r \equiv 1)$ of \scup{1-DNR}.
        \begin{enumerate}
			\item Compute an optimal splittable flow $f_\mathrm{split}$ solving the quadratic program corresponding to $\III$.
			\item Use randomized rounding to obtain a confluent flow $f_\mathrm{conf}$ from $f_\mathrm{split}$: for each vertex with positive outflow, choose one outgoing arc with probability proportional to the flow it carries.
			\item Return $f_\mathrm{conf}$.
		\end{enumerate}
    \end{algo}
\end{framed}

We regard the output of the algorithm as a random confluent flow $F_\mathrm{conf}$, where the randomness is in the selection of the outgoing arcs.
We claim that the rounding degrades the quality by a factor of at most $\sqrt{n}$ in expectation.
Formally, we argue that
\begin{displaymath}
    \EE \big[ L(F_\mathrm{conf}) \big]
    \leq \sqrt{n} \sum_{x \in V} \big( f_\mathrm{split}^\mathrm{out}(x) \big)^2
    \leq \sqrt{n} \OPT(\III),
\end{displaymath}
where the second inequality follows immediately because $f_\mathrm{split}$ solves the QP relaxation of~$\III$.
It remains to show the first inequality.

For each $x \in V'$, let $P_x$ denote the random path from $x$ to the sink $t$ formed by the selected outgoing arcs.
Thus, $F_\mathrm{conf}$ and $P_x$ are determined by the same random choices and thereby generally dependent.
We may interpret $P_x$ as a random walk taken by a unit of flow starting at vertex $x$.
We will use the linearity of expectation to decompose the expected loss $\EE \big[ L(F_\mathrm{conf}) \big]$ into pairwise contributions arising from dependent random paths.

For $v,x \in V'$, the indicator variable $\bm{1} [v \in P_x]$ is $1$ if $P_x$ visits $v$, and $0$ otherwise.
Using this notation, we have
\begin{displaymath}
    F_\mathrm{conf}^\mathrm{out}(v) = \sum_{x \in S'} -d(x) \bm{1} [v \in P_x]
\end{displaymath}
since each vertex $x \in S'$ contributes $-d(x)$ units of flow whenever $P_x$ visits $v$.
Squaring yields
\begin{align*}
    \big( F_\mathrm{conf}^\mathrm{out}(v) \big)^2
    &= \left( \sum_{x \in S'} -d(x) \bm{1} [v \in P_x] \right) \left( \sum_{x \in S'} -d(x) \bm{1} [v \in P_x] \right) \\
    &= \sum_{x \in S'} \sum_{y \in S'} d(x) d(y) \bm{1} [v \in P_x] \bm{1} [v \in P_y].
\end{align*}
This allows us to rewrite the expected loss as
\begin{align*}
    \EE [ L(F_\mathrm{conf}) ]
    &= \EE \left[ \sum_{v \in S'} \big( F_\mathrm{conf}^\mathrm{out}(v) \big)^2 \right] \\
    &= \EE \left[ \sum_{v \in S'} \sum_{x \in S'} \sum_{y \in S'} d(x)d(y) \bm{1} [v \in P_x] \bm{1} [v \in P_y] \right] \\
    &= \EE \left[ \sum_{x \in S'} \sum_{y \in S'} d(x)d(y) \sum_{v \in S'} \bm{1} [v \in P_x] \bm{1} [v \in P_y] \right] \\
    &= \sum_{x \in S'} \sum_{y \in S'} d(x)d(y) \EE \left[ \sum_{v \in S'} \bm{1} [v \in P_x] \bm{1} [v \in P_y] \right],
\end{align*}
where the last equality is due to the linearity of expectation.

Note that $\EE \left[ \sum_{v \in S'} \bm{1} [v \in P_x] \bm{1} [v \in P_y] \right]$ is the expected length of the overlap between~$P_x$ and $P_y$.
The challenge is that $P_x$ and $P_y$ are generally dependent.
We overcome this by conditioning on the first meeting point.
For all $x,y \in S'$, let $M_{x,y}$ be a random variable representing the vertex $v \in V'$ at which $P_x$ and $P_y$ meet for the first time.
If two random paths meet at the same vertex $v$, then they take the same route from $v$ onward; this is because each vertex of the rounded flow has at most one outgoing arc.
We denote the (random) length of a (random) path $P$ by $|P|$.
Using the law of total expectation, we get
\begin{align*}
    \EE \left[ \sum_{v \in S'} \bm{1} [v \in P_x] \bm{1} [v \in P_y] \right]
    &= \sum_{u \in V'} \PP(M_{x,y} = u) \EE \left[ \sum_{v \in S'} \bm{1} [v \in P_x] \bm{1} [v \in P_y] \middle| M_{x,y} = u \right] \\
    &= \sum_{u \in V'} \PP(M_{x,y} = u) \EE \left[ \sum_{v \in S'} \bm{1} [v \in P_u] \right] \\
    &= \sum_{u \in V'} \PP(M_{x,y} = u) \EE [ | P_u | ]
\end{align*}

We have
\begin{displaymath}
    \EE [ | P_u | ]
    = \EE \left[ \sum_{v \in S'} \bm{1} [v \in P_u] \right]
    = \sum_{v \in S'} \PP (v \in P_u).
\end{displaymath}
Squaring yields
\begin{displaymath}
    \EE [ | P_u | ]^2
    = \left( \sum_{v \in S'} \PP (v \in P_u) \right)^2
    \leq n \sum_{v \in S'} \PP (v \in P_u)^2,
\end{displaymath}
where the inequality is due to Cauchy-Schwarz.
By taking the square root, we get
\begin{displaymath}
    \EE [ | P_u | ]
    \leq \sqrt{n} \sqrt{\sum_{v \in S'} \PP (v \in P_u)^2}
    \leq \sqrt{n} \sum_{v \in S'} \PP (v \in P_u)^2,
\end{displaymath}
where the second inequality holds because of $\sum_{v \in S'} \PP (v \in P_u)^2 \geq \PP (u \in P_u)^2 = 1$ if $u \in S'$, and because of $\EE [ | P_t | ] = 0$ if $u = t$.

For each $x \in V'$, let $P_x'$ and $P_x''$ be independently sampled copies of $P_x$; they are independent of all other random paths and of each other.
We find that
\begin{align*}
    \EE [ | P_u | ]
    &\leq \sqrt{n} \sum_{v \in S'} \PP (v \in P_u)^2 \\
    &= \sqrt{n} \sum_{v \in S'} \PP (v \in P_u') \PP (v \in P_u'') \\
    &= \sqrt{n} \sum_{v \in S'} \EE \big[ \bm{1} [v \in P_u'] \big] \EE \big[ \bm{1} [v \in P_u''] \big] \\
    &= \sqrt{n} \EE \left[ \sum_{v \in S'} \bm{1} [v \in P_u'] \bm{1} [v \in P_u''] \right],
\end{align*}
where the last equality uses the independence of $P_u'$ and $P_u''$.

For $x,y \in S'$, let $M_{x,y}'$ be a random variable representing the vertex at which $P_x'$ and $P_y''$ first meet.
We claim that $\Pr(M_{x,y} = u) = \Pr(M'_{x,y} = u)$ for all $x,y \in S'$ and $u \in V'$.
The random paths $P_x$ and $P_y$ are dependent because they take the same route after they first meet.
The decision which outgoing arc to choose, however, is made independently at random for each vertex.
Hence, until their first meeting, $P_x$ and $P_y$ evolve independently, making their joint behavior indistinguishable from the pair $P_x'$, $P_y''$.
This shows that the claim holds.

We now derive a useful upper bound for the expected overlap of $P_x$ and $P_y$.
\begin{align*}
    \EE \left[ \sum_{v \in S'} \bm{1} [v \in P_x] \bm{1} [v \in P_y] \right]
    &= \sum_{u \in V'} \PP(M_{x,y} = u) \EE [ | P_u | ] \\
    &\leq \sqrt{n}\! \sum_{u \in V'} \PP(M_{x,y} = u) \EE \left[ \sum_{v \in S'} \bm{1} [v \in P_u'] \bm{1} [v \in P_u''] \right] \\
    &= \sqrt{n}\! \sum_{u \in V'} \PP (M_{x,y}' = u) \EE \left[ \sum_{v \in S'} \bm{1} [v \in P_x'] \bm{1} [v \in P_y''] \middle| M_{x,y}' = u \right] \\
    &= \sqrt{n} \EE \left[ \sum_{v \in S'} \bm{1} [v \in P_x'] \bm{1} [v \in P_y''] \right]
\end{align*}

To complete the proof, we need the following lemma, which states that the expected outflow of $F_\mathrm{conf}$ matches the outflow of $f_\mathrm{split}$ at each vertex.
\begin{lemma}\label{lem:splittable-flow}
    For each $v \in S'$, we have $\EE [F_\mathrm{conf}^\mathrm{out}(v)] = f_\mathrm{split}^\mathrm{out}(v)$.
\end{lemma}

The proof of the lemma is by induction over a topological ordering of the vertices of $G'$; we defer it to the end of this section.
Finally, we find that
\begin{align*}
    \EE [L(F_\mathrm{conf})]
    &= \sum_{x \in S'} \sum_{y \in S'} d(x)d(y) \EE \left[ \sum_{v \in S'} \bm{1} [v \in P_x] \bm{1} [v \in P_y] \right] \\
    &\leq \sum_{x \in S'} \sum_{y \in S'} d(x)d(y) \sqrt{n} \EE \left[ \sum_{v \in S'} \bm{1} [v \in P_x'] \bm{1} [v \in P_y''] \right] \\
    &= \sqrt{n} \sum_{v \in S'} \EE \left[ \sum_{x \in S'\vphantom{y}} -d(x) \bm{1} [v \in P_x'] \right] \EE \left[ \sum_{y \in S'} -d(y) \bm{1} [v \in P_y''] \right] \\
    &= \sqrt{n} \sum_{v \in S'} \big( \EE [F_\mathrm{conf}^\mathrm{out}(v)] \big)^2 \\
    &= \sqrt{n} \sum_{v \in S'} \big( f_\mathrm{split}^\mathrm{out} (v) \big)^2 \\
    &\leq \sqrt{n} \OPT(\III).
\end{align*}

We conclude that \cref{alg:single} is a randomized $\bigO(\sqrt{n})$-approximation algorithm for \scup{1-DNR}.

\begin{proof}[Proof of \cref{lem:splittable-flow}]
    We use induction along a topological ordering of the vertices of $G'$.
    For the base case, let $v \in S'$ be a vertex without incoming arcs.
    Then, the outflow of $v$ equals~$-d(v)$, and hence $\EE [F_\mathrm{conf}^\mathrm{out}(v)] = -d(v) = f_\mathrm{split}^\mathrm{out}(v)$.

    We proceed to the induction step.
    Let $v \in S'$ and assume that the induction hypothesis holds for all predecessors of $v$.
    We find that
    \begin{align*}
        \EE [F_\mathrm{conf}^\mathrm{out}(v)]
        &= -d(v) + \sum_{u \colon (u,v) \in A'} \EE [F_\mathrm{conf}^\mathrm{out}(u)] p_u(v) \\
        &= -d(v) + \sum_{u \colon (u,v) \in A'} f_\mathrm{split}^\mathrm{out}(u) p_u(v) \\
        &= -d(v) + \sum_{u \colon (u,v) \in A'} f_\mathrm{split}(u,v) \\
        &= f_\mathrm{split}^\mathrm{out}(v).\qedhere
    \end{align*}
\end{proof}

\subsection{Single-Source Hardness}\label{sub:single-hardness}

In this section, we prove the following.

\begin{theorem}\label{thm:1dnr-hardness}
    \scup{1-DNR} is $\mathrm{APX}$-hard, even if each resistance is either $0$ or $1$ and all demands are polynomially bounded in the instance size.
\end{theorem}

We perform an L-reduction from the following problem.

\probDefOptimization{Quadratic Load Balancing (QLB)}
{A set $J$ of jobs, a set $M$ of machines, job weights $w \colon J \to \NN$, and a mapping $\mathcal{A} \colon J \to 2^M$ that assigns to each job $j \in J$ a nonempty set $\mathcal{A}(j)$ of available machines.}
{An assignment $\sigma \colon J \to M$ with $\sigma(j) \in \mathcal{A}(j)$ for every $j \in J$.}
{Minimize the cost $\displaystyle C(\sigma) \ceq \sum_{i \in M} \left( \sum_{j \colon \sigma(j) = i} w(j) \right)^2.$
}

The \scup{QLB} problem is $\mathrm{APX}$-hard even if each job weight is either $1$ or $3$ \cite{azar2004all}.
For an L-reduction, we must provide two polynomial-time mappings.
First, a ``forward'' mapping that maps a \scup{QLB} instance $\III$ to a \scup{1-DNR} instance $\III'$ such that $\OPT(\III') \leq \alpha \OPT(\III)$ for some constant $\alpha$.
Second, a ``backward'' mapping that maps a solution $T$ for $\III'$ to a solution $\sigma$ for $\III$ such that $C(\sigma) - \OPT(\III) \leq \beta (L(T) - \OPT(\III'))$ for some constant $\beta$.
We will prove these inequalities for $\alpha = \beta = 1$.

\begin{construction}[Forward Mapping]\label{cons:1dnr-forward}
    Let $\III = (J, M, w, \mathcal{A})$ be an instance of \scup{QLB}.
    We construct an instance $\III' = (G=(V,E), d, r)$ of \scup{1-DNR} as follows (see \cref{fig:single-source-reduction} for an illustration).

    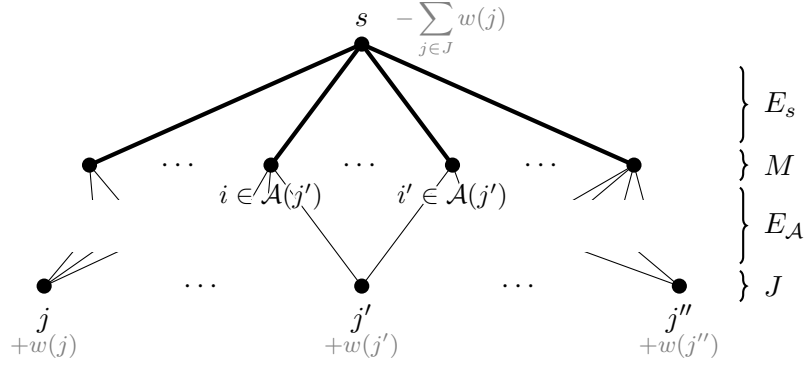
\begin{figure}[t]
    \centering

    \begin{tikzpicture}[
            yscale=-1,
            circ/.style={fill,circle,inner sep=0pt,minimum size=2mm}
        ]

        \tikzmath{
            \w = 1.2;
            \h = 1.6;
            \margin = 5;
            \sgap = 0.2;
            \lgap = 0.3;
            \stub = 0.45;
        }
        \contourlength{0.3mm}

        \node[
            style=circ,
            label={north:$s$},
            label={[xshift=2mm,yshift=2mm]east:\footnotesize\color{gray} $\displaystyle -\sum_{j \in J} w(j)$}
        ] (s) at (0,0) {};

        \node[
            style=circ
        ] (m1) at (-3*\w,\h) {};

        \node at (-2*\w,\h) {$\cdots$};

        \node[
            style=circ,
            label={south:\small\contour*{white}{$i \in \mathcal{A}(j')$}}
        ] (m2) at (-1*\w,\h) {};

        \node at (0*\w,\h) {$\cdots$};

        \node[
            style=circ,
            label={south:\small\contour*{white}{$i' \in \mathcal{A}(j')$}}
        ] (m3) at (1*\w,\h) {};

        \node at (2*\w,\h) {$\cdots$};

        \node[
            style=circ
        ] (m4) at (3*\w,\h) {};

        \node[
            style=circ,
            label={south:$j$\vphantom{$j'$}},
            label={[yshift=-1em]south:\footnotesize\color{gray}$+w(j)$\vphantom{$w(j')$}}
        ] (j1) at (-3.5*\w,2*\h) {};

        \node at (-1.75*\w,2*\h) {$\cdots$};

        \node[
            style=circ,label={south:$j'$},
            label={[yshift=-1em]south:\footnotesize\color{gray}$+w(j')$}
        ] (j2) at (0,2*\h) {};

        \node at (1.75*\w,2*\h) {$\cdots$};

        \node[style=circ,label={south:$j''$},
            label={[yshift=-1em]south:\footnotesize\color{gray}$+w(j'')$}
        ] (j3) at (3.5*\w,2*\h) {};

        \begin{scope}[on background layer]
            \begin{scope}[overlay]
                \clip (-5*\w,\h-\stub) rectangle (5*\w,\h+\stub);

                \draw (m1.center) -- ++(85:2);
                \draw (m1.center) -- ++(50:2);

                \draw (m2.center) -- ++(120:2);
                \draw (m2.center) -- ++(95:2);

                \draw (m3.center) -- ++ (70:2);

                \draw (m4.center) -- ++ (150:2);
                \draw (m4.center) -- ++ (140:2);
                \draw (m4.center) -- ++ (105:2);
                \draw (m4.center) -- ++ (75:2);
            \end{scope}

            \begin{scope}[overlay]
                \clip (-5*\w,2*\h-\stub) rectangle (5*\w,2*\h+\stub);

                \draw (j1.center) -- ++(-25:2);
                \draw (j1.center) -- ++(-35:2);
                \draw (j1.center) -- ++(-50:2);

                \draw (j3.center) -- ++ (-140:2);
                \draw (j3.center) -- ++ (-160:2);
            \end{scope}

            \draw[ultra thick] (s.center) -- (m1.center);
            \draw[ultra thick] (s.center) -- (m2.center);
            \draw[ultra thick] (s.center) -- (m3.center);
            \draw[ultra thick] (s.center) -- (m4.center);
            \draw (m2.center) -- (j2.center);
            \draw (m3.center) -- (j2.center);
        \end{scope}

        \begin{scope}[overlay]
            \draw[thick,decoration={brace},decorate]
            (\margin,\lgap) -- node[right=5pt] {$E_s$} (\margin,\h-\lgap);

            \draw[thick,decoration={brace},decorate]
            (\margin,\h-\sgap) -- node[right=5pt] {$M$} (\margin,\h+\sgap);

            \draw[thick,decoration={brace},decorate]
            (\margin,\h+\lgap) -- node[right=5pt] {$E_\mathcal{A}$} (\margin,2*\h-\lgap);

            \draw[thick,decoration={brace},decorate]
            (\margin,2*\h-\sgap) -- node[right=5pt] {$J$} (\margin,2*\h+\sgap);
        \end{scope}
    \end{tikzpicture}

    \caption{%
        Illustration of \cref{cons:1dnr-forward}.
        Demands are shown in gray.
        Thick edges have unit resistance, while thin edges have zero resistance.
        The illustration shows example machines~$i$ and $i'$, as well as example jobs $j$, $j'$, and $j''$.
        The job $j'$ can be processed on either machine~$i$ or $i'$.
    }
    \label{fig:single-source-reduction}
\end{figure}

    Let $V \ceq \{s\} \cup J \cup M$.
    Let $E_s \ceq \{ \{s,i\} \mid  i \in M\}$, $E_\mathcal{A} \ceq \{ \{i,j\} \mid j \in J,\, i \in \mathcal{A}(j) \}$, and $E \ceq E_s \cup E_\mathcal{A}$.
    Moreover, let $d(j) \ceq w(j)$ for each $j \in J$, $d(i) \ceq 0$ for each $i \in M$, and $d(s) \ceq - \sum_{j \in J} w(j)$.
    Finally, let $r(e) \ceq 1$ if $e \in E_s$, and $r(e) \ceq 0$ otherwise.
\end{construction}

Note that $G$ is connected and that $\sum_{v \in V} d(v) = 0$.
Moreover, all demands are polynomially bounded in the instance size.
Let $\III$ be an instance of \scup{QLB}, and let $\III'$ be the corresponding instance of \scup{1-DNR}.
We show that $\OPT(\III') \leq \OPT(\III)$.
Let $\sigma^* \colon J \to M$ be an optimal assignment for $\III$.
Note that $T^* \ceq (V, E^*)$ with $E^* \ceq E_s \cup \{ \{i,j\} \mid j \in J,\, i = \sigma^*(j) \}$ is a spanning tree of $G$.
We find that
\begin{align*}
    L(T^*)
    &= \sum_{\{s,i\} \in E_s} \big( f_{T^*}(s,i) \big)^2
    = \sum_{\{s,i\} \in E_s} \left( \sum_{\{i,j\} \in E^* \cap E_\mathcal{A}} d(j) \right)^2 \\
    &= \sum_{i \in M} \left( \sum_{j \colon \sigma^*(j) = i} w(j) \right)^2
    = C(\sigma^*).
\end{align*}
Hence, $\OPT(\III') \leq L(T^*) = C(\sigma^*) = \OPT(\III)$.

\begin{construction}[Backward Mapping]
    Let $T \subseteq G$ be a solution to $\III'$.
    We construct a solution $\sigma \colon J \to M$ to $\III$ as follows.
    For each $j \in J$ there exists a unique $(j,s)$-path in $T$; let $\{j, i_j\}$ be the first edge on this path.
    Then, let $\sigma(j) \ceq i_j$ for each $j \in J$.
\end{construction}

Let $T = (V,E') \subseteq G$ be a solution to $\III'$, and let $\sigma \colon J \to M$ be the corresponding solution to $\III$.
We first show that $C(\sigma) \leq L(T)$.
Consider an edge $\{s,i\} \in E_s \cap E(T)$.
The removal of $\{s,i\}$ splits $T$ into two components; let $V_i$ be the vertex set of the component containing~$i$.
We get
\begin{align*}
    C(\sigma)
    &= \sum_{i \in M} \left( \sum_{j \colon \sigma(j) = i} w(j) \right)^2
    = \sum_{\{s,i\} \in E'} \, \sum_{h \in V_i \cap M} \left( \sum_{j \colon \sigma(j) = h} w(j) \right)^2 \\
    &\leq \sum_{\{s,i\} \in E'} \left( \sum_{h \in V_i \cap M} \,\, \sum_{j \colon \sigma(j) = h} w(j) \right)^2
    = \sum_{\{s,i\} \in E'} \left( \sum_{j \in V_i \cap J} w(j) \right)^2 \\
    &= \sum_{\{s,i\} \in E'} \left( f_T(s,i) \right)^2
    = L(T),
\end{align*}
where the second equality follows since each machine appears in exactly one component of $T - \{s\}$, and the third-to-last equality follows because $\sigma$ assigns each job to exactly one machine.
Using $\OPT(\III') \leq \OPT(\III)$, we conclude that $C(\sigma) - \OPT(\III) \leq L(T) - \OPT(\III')$.
Therefore, the mappings satisfy the required inequalities and \cref{thm:1dnr-hardness} holds.

\section{Multi-Source Networks}\label{sec:multi}

\subsection{A Multi-Source Approximation Algorithm}\label{sub:multi-approx}

We provide an $(n-1)$-approximation algorithm for \scup{DNR} by a relaxation to the \scup{Transshipment} problem, also known as \scup{Uncapacitated Minimum Cost Flow}.

\probDefOptimization{Transshipment}
{A connected undirected graph $G = (V,E)$, a demand function $d \colon V \to \RR$ satisfying $\sum_{x \in V} d(x) = 0$, and a cost function $c \colon E \to \RR_{\geq 0}$.}
{A flow $f$ satisfying $d$.}
{Minimize $\sum_{\{x,y\} \in E} c(\{x,y\}) \cdot \big| f(x,y) \big|$.}

It is well-known that every instance of the \scup{Transshipment} problem has an optimal solution that induces a forest and can be found in polynomial time \cite[Chapter~11.2]{ahuja1993network}.
Hence, the relevant difference between \scup{Transshipment} and \scup{DNR} lies in their objective functions: edge flows appear linearly in the former and quadratically in the latter.
We claim that the following algorithm is an $(n-1)$-approximation algorithm for \scup{DNR}.

\begin{framed}
    \begin{algo}\label{alg:general}
        \emph{Input:} An instance $\III = (G=(V,E), d, r)$ of \scup{DNR}.
        \begin{enumerate}
            \item Construct an instance $\III' = (G=(V,E), d, c)$ of \scup{Transshipment} that has $G$ and $d$ in common with $\III$ and has $c(e) \ceq \sqrt{r(e)}$ for each $e \in E$.
            \item Compute an optimal acyclic flow solution $f$ to $\III'$.
            \item Return $f$.
        \end{enumerate}
    \end{algo}
\end{framed}

\begin{theorem}\label{thm:general-tractability}
    \cref{alg:general} is an $(n-1)$-approximation algorithm for \scup{DNR}.
\end{theorem}
\begin{proof}
    Let $f$ be any acyclic flow over $G$ satisfying $d$.
    For every edge $\{x,y\} \in E$, define the auxiliary quantity $w_f(\{x,y\}) \ceq \sqrt{r(\{x,y\})} \cdot |f(x,y)|$.
    The cost of $f$ in the context of \scup{Transshipment} is $C(f) \ceq \sum_{e \in E} w_f(e)$, whereas the power loss is $L(f) = \sum_{\{x,y\} \in E} r(\{x,y\}) \cdot f(x,y)^2 = \sum_{e \in E} w_f(e)^2$.
    Since $f$ is acyclic, at most $n-1$ values of $w_f$ are nonzero.
    Hence, we have
    \begin{displaymath}
        \sum_{e \in E} w_f(e)^2
        \leq \left( \sum_{e \in E} w_f(e) \right)^2
        \leq (n - 1) \sum_{e \in E} w_f(e)^2,
    \end{displaymath}
    where the first inequality follows from the nonnegativity of $w_f$ and the second inequality is due to Cauchy-Schwarz.
    Equivalently, $L(f) \leq C(f)^2 \leq (n-1) \cdot L(f)$.
    Let $\hat{f}$ be an optimal acyclic \scup{Transshipment} flow, and let $f^*$ be an optimal \scup{DNR} flow.
    We find that $L(\hat{f}) \leq C(\hat{f})^2 \leq C(f^*)^2 \leq (n-1) \cdot L(f^*)$.
    Therefore, \cref{alg:general} is an ($n - 1$)-approximation algorithm.
\end{proof}

\subsection{Multi-Source Hardness}

We prove two inapproximability results for multi-source networks.
First, on planar graphs with an unrestricted number of sources, \scup{DNR} cannot be approximated within a factor of $n^{1-\eps}$ for any $\eps > 0$ unless $\mathrm{P} = \mathrm{NP}$.
We obtain this result by a reduction from a special case of \scup{Perfectly Balanced Connected Partition}, which we prove NP-hard beforehand.
Second, we consider networks containing only two sources: on general graphs, \scup{2-DNR} cannot be approximated within a factor of $c \log^2 n$ for some constant $c > 0$ unless $\mathrm{P} = \mathrm{NP}$.

\subsubsection{Hardness of Perfectly Balanced Connected Partition}\label{sub:bcp-hardness}

We prove a conjecture of \textcite{dyer1985complexity} on the complexity of partitioning the vertex set of a graph into equally sized connected parts.

\probDefDecision{Perfectly Balanced Connected Partition (BCP)}
{A graph $G = (V,E)$ and a positive integer $k$.}
{Is there a partition of $V$ into sets $V_1, \dots, V_k$ such that the induced subgraph $G[V_i]$ is connected and $|V_i| = |V| / k$ for each $i \in \{1, \dots, k\}$?}

\textcite{dyer1985complexity} established two separate hardness results for this problem.
First, they proved NP-hardness on planar graphs.
Second, they proved NP-hardness for every fixed number $k \geq 2$ of parts.
They conjectured that the problem remains NP-hard even when these two restrictions are imposed simultaneously.
Subsequent work continues to cite the hardness on planar graphs and for fixed $k \geq 2$ as separate facts \cite{enciso2009what,validi2022political,blavzej2024equitable}, suggesting that the conjecture is still open.
When vertices have weights, the problem is known to be weakly NP-hard for every fixed $k \geq 2$ on grid graphs of height three \cite{becker1998max} and on series-parallel graphs \cite{ito2006partitioning}.
We first prove the conjecture for $k = 2$.
The full conjecture then follows as a corollary and implies strong NP-hardness for the vertex-weighted setting.

\begin{theorem}\label{thm:bcp}
    \scup{BCP} is NP-hard on planar graphs with $k = 2$.
\end{theorem}

We perform a polynomial-time reduction from a variant of the following problem, which asks whether a partial 2-coloring of a graph can be extended so that the vertices of each of the two colors induce a connected subgraph.

\probDefDecision{2-Disjoint Connected Subgraphs (2DCS)}
{A graph $G = (V,E)$ and disjoint sets $R,B \subseteq V$.}
{Are there disjoint sets $R', B' \subseteq V$ such that $R' \supseteq R$, $B' \supseteq B$, and both $G[R']$ and $G[B']$ are connected?}

We refer to the vertices in $R$ as red, the vertices in $B$ as blue, and the remaining vertices as white.
\textcite{gray2012removing} prove that \scup{2DCS} is NP-hard on planar graphs.
By inspecting their reduction, we find that the constructed instance is a connected graph with a planar embedding in which the boundary of some face contains both a red and a blue vertex.
(Each ``inverter gadget'' in their proof has such a face.)

\begin{fact}[\cite{gray2012removing}]
    \scup{2DCS} is NP-hard on connected planar graphs, even if each instance comes with a planar embedding containing a face whose boundary includes both a red and a blue vertex.
\end{fact}

We use this as the source problem in our reduction.
The idea is to copy the graph of the \scup{2DCS} instance and simulate vertex weights by attaching pendant vertices (vertices of degree one) to red and blue vertices.
Red vertices receive many pendants, most blue vertices receive a moderate number of pendants, and white vertices receive none.
A distinguished red vertex $v_r$ and a distinguished blue vertex $v_b$ receive specially chosen numbers of pendants so that the total pendant weight of the red vertices equals that of the blue vertices.
This will force any feasible partition to separate the red and blue vertices.

\begin{construction}\label{cons:bcp}
    Let $\III = (G,R,B)$ be an instance of \scup{2DCS} where $G$ is a connected planar graph.
    Fix a planar embedding of $G$ with a face $f$ whose boundary contains both a red vertex $v_r \in R$ and a blue vertex $v_b \in B$.
    Construct an instance $\III' = (G',k)$ of \scup{BCP} with $k = 2$ as follows.

    Let $G'$ initially be a copy of $G$, and let $n \ceq |V(G)|$.
    Attach $100n^2$ pendant vertices to each vertex in $R \setminus \{v_r\}$.
    Attach $10n$ pendant vertices to each vertex in $B \setminus \{v_b\}$.
    Attach $100n^2 + |B| \cdot 10n$ pendant vertices to $v_r$, and $|R| \cdot 100n^2 + 10n$ pendant vertices to $v_b$.
    Finally, add a new $(v_r, v_b)$-path $P_{rb}$ with $n$ internal vertices, drawn in the interior of $f$.
\end{construction}

The graph $G'$ is planar by construction and has $|R| \cdot 200n^2 + |B| \cdot 20n + 2n$ vertices.
Hence, for any solution $V_1, V_2$ to $\III'$, we have $|V_1| = |V_2| = |R| \cdot 100n^2 + |B| \cdot 10n + n$.

For each original vertex $v \in V(G)$, let $\operatorname{pen}(v) \subseteq V(G')$ denote the pendant vertices attached to $v$ in the construction, and let $\operatorname{pen}[v] \ceq \operatorname{pen}(v) \cup \{v\}$.

\begin{observation}
    Let $V_1,V_2$ be any solution to $\III'$.
    For each vertex $v \in V(G)$, we have either $\operatorname{pen}[v] \subseteq V_1$ or $\operatorname{pen}[v] \subseteq V_2$.
\end{observation}
\begin{proof}
    Assume for contradiction that $v \in V_1$ and some pendant vertex $p \in \operatorname{pen}(v)$ is contained in $V_2$.
    Then, $p$ is isolated in $G'[V_2]$.
    Since $|V_2| > 1$, we conclude that $G'[V_2]$ is not connected, a contradiction.
\end{proof}

\begin{lemma}
    Instance $\III = (G,R,B)$ of \scup{2DCS} is a \ttup{yes}-instance if and only if the instance $\III' = (G',k=2)$ of \scup{BCP} obtained from $\III$ is a \ttup{yes}-instance.
\end{lemma}
\begin{proof}
    ($\Rightarrow$)
    Let $R',B'$ be a solution to $\III$ and arbitrarily extend $R'$ and $B'$ to a partition $R'',B''$ of $V(G)$ such that $G[R'']$ and $G[B'']$ are connected.
    We construct a partition $V_1, V_2$ of $V(G')$.
    Initially, let $V_1 \ceq \bigcup_{v \in R''} \operatorname{pen}[v]$ and $V_2 \ceq \bigcup_{v \in B''} \operatorname{pen}[v]$.
    This leaves only the interior vertices of the path $P_{rb}$ unassigned.
    We assign to $V_1$ the $|B''|$ interior vertices of $P_{rb}$ that are closest to $v_r$, and to $V_2$ the $|R''|$ interior vertices of $P_{rb}$ that are closest to $v_b$.
    Note that both $G'[V_1]$ and $G'[V_2]$ are connected.
    Moreover, $V_1$ contains $|B''|$ vertices of $P_{rb}$, $|R''|$ vertices of $G$, and all pendants attached to red vertices.
    Hence, we have
    \begin{align*}
        |V_1|
        &= |R''| + |B''| + (|R| - 1) \cdot 100n^2 + 100n^2 + |B| \cdot 10n \\
        &= n + |R| \cdot 100n^2 + |B| \cdot 10n \\
        &= |V(G')|/2.
    \end{align*}
    Since $V_1, V_2$ is a partition of $V(G')$, this implies that $|V_1| = |V_2|$.
    Thus, $V_1, V_2$ is a solution to $\III'$.

    ($\Leftarrow$)
    Let $V_1, V_2$ be a solution to $\III'$ and assume without loss of generality (by symmetry) that $v_b \in V_2$.
    First, we show that every red vertex is contained in $V_1$.
    Suppose that there is a red vertex $v \in R$ with $v \in V_2$.
    Then, we obtain the contradiction
    \begin{displaymath}
        |V_2|
        \geq |\operatorname{pen}[v_b]| + |\operatorname{pen}[v]|
        > |R| \cdot 100n^2 + 100n^2
        > |R| \cdot 100n^2 + |B| \cdot 10n + n
        = |V_2|,
    \end{displaymath}
    where we use $|B| \leq n$ in the third inequality.
    
    Next, we show that every blue vertex is contained in $V_2$.
    Assume for contradiction that there is a blue vertex $v \in B$ with $v \notin V_2$.
    Hence, apart from $v_b$, at most $|B| - 2$ blue vertices are contained in $V_2$.
    Moreover, $G'$ contains at most $2n$ non-pendant vertices, and recall that $V_2$ contains no red vertices.
    We obtain the contradiction
    \begin{align*}
        |V_2|
        &\leq |\operatorname{pen}(v_b)| + (|B| - 2) \cdot 10n + 2n \\
        &= |R| \cdot 100n^2 + (|B| - 1) \cdot 10n + 2n \\
        &< |R| \cdot 100n^2 + |B| \cdot 10n + n \\
        &= |V_2|.
    \end{align*}
    
    In summary, every red vertex lies in $V_1$, and every blue vertex lies in $V_2$.
    Now, let $R' \ceq V_1 \cap V(G)$ and $B' \ceq V_2 \cap V(G)$, and note that $R', B'$ is a solution to $\III$.
\end{proof}

This concludes the proof of \cref{thm:bcp}.
The NP-hardness for every fixed $k \geq 2$ follows by a simple reduction from the $k = 2$ case first used by \textcite{dyer1985complexity}.
We sketch it here.

\begin{corollary}
    \scup{BCP} is NP-hard on planar graphs for every fixed $k \geq 2$.
\end{corollary}
\begin{proof}
    Given an instance $(G,2)$ of \scup{BCP} with $n \ceq |V(G)|$, attach a path consisting of $(k - 2) \cdot \lfloor n / 2 \rfloor$ new vertices to some vertex of $G$.
    Let $G'$ be the resulting graph.
    Then, any feasible partition of $G'$ into $k$ equally sized parts must consist of two parts dividing the original graph $G$ and $k - 2$ parts dividing the new path.
    Hence, $(G',k)$ is equivalent to $(G,2)$.
\end{proof}

This resolves the conjecture of \textcite{dyer1985complexity}.
For the reduction to \scup{DNR} in the next section, we consider an extension of \scup{BCP} in which each part must contain a prescribed terminal vertex: the input contains pairwise distinct terminal vertices $v_1, \dots, v_k$, and a feasible partition $V_1, \dots, V_k$ must satisfy $v_i \in V_i$ for each $i \in \{1, \dots, k\}$.
We obtain the following corollary.

\begin{corollary}\label{cor:bcp-terminals}
    The prescribed-terminal variant of \scup{BCP} is NP-hard on connected planar graphs with $k = 2$, even if the two terminals lie on the boundary of the outer face.
\end{corollary}
\begin{proof}
    We have seen that every solution to the instance obtained by \cref{cons:bcp} separates the red and the blue vertices.
    Hence, the problem remains NP-hard if we prescribe $v_r$ and $v_b$ as terminals.
    The constructed graph is planar and connected, and since $v_r$ and $v_b$ lie on the boundary of a common face even after adding the path $P_{rb}$, we may choose an embedding where this face is the outer face.
\end{proof}

\subsubsection{Hardness of Multi-Source Distribution Network Reconfiguration}\label{sub:multi-hardness}

Next, we show that, for every $\eps > 0$, \scup{DNR} on planar graphs cannot be approximated within a factor of $n^{1-\eps}$ unless $\mathrm{P} = \mathrm{NP}$.
First, we prove the following somewhat weaker theorem, which establishes inapproximability for every constant approximation factor.
We will later derive the stronger statement as a corollary.

\begin{theorem}\label{thm:constant-factor-hardness}
    For every $\rho > 1$, there is no $\rho$-approximation algorithm for \scup{DNR} on planar graphs unless $\mathrm{P} = \mathrm{NP}$.
    This holds even if each resistance is either $0$ or $1$ and all demands are polynomially bounded in the instance size.
\end{theorem}

We perform a gap-introducing reduction from the prescribed-terminal variant of \scup{BCP} on connected planar graphs with the two terminals on the boundary of the outer face (see \cref{cor:bcp-terminals}).
A gap-introducing reduction maps \ttup{yes}-instances to instances with small optimum value and \ttup{no}-instances to instances with large optimum value in polynomial time.
An approximation algorithm that bridges the gap would distinguish the two cases.
We may assume that the number of vertices in the \scup{BCP} instance is even since this only discards trivial \ttup{no}-instances.

\begin{construction}\label{cons:general-networks}
    Fix $\rho > 1$ and let $\III = (G, x, y)$ be an instance of \scup{BCP} where $G$ is connected, planar, and has prescribed terminal vertices $x, y$ on the boundary of its outer face.
    Construct an instance $\III' = (G', d, r)$ of \scup{DNR} as follows (see \cref{fig:general-network-reduction} for an illustration).

    \begin{figure}[t]
    \centering

    \begin{tikzpicture}[
            yscale=-1,
            circ/.style={fill,circle,inner sep=0pt,minimum size=2mm},
            graph/.style={draw,dashed,ellipse,thick,fill=gray!10}
        ]

        \tikzmath{
            \w = 2.7;
            \h = 1.5;
            \wgap = 3;
            \hgap = 1.1;
            \hbot = 1.2;
            \epsdist = 0.15;
        }
        \contourlength{0.3mm}

        \node[
            style=circ,
            label={north:$a$},
            label={south:\small\color{gray}$-\rho$}
        ] (a) at (2*\w,-\h) {};

        \node[
            style=circ,
            label={west:$b_1$},
            label={east:\footnotesize\color{gray}\contour*{white}{$\displaystyle -\frac{\rho N}{2} + 1$}}
        ] (b1) at (0,0) {};

        \node[
            style=circ,
            label={west:$b_2$},
            label={east:\footnotesize\color{gray} $\displaystyle -\frac{\rho N}{2} + 1$}
        ] (b2) at (\w,0) {};

        \node at (2*\w,0) {$\cdots$};

        \node[
            style=circ,
            label={west:$b_\rho$},
            label={east:\footnotesize\color{gray} $\displaystyle -\frac{\rho N}{2} + 1$}
        ] (br) at (3*\w,0) {};

        \node[
            style=circ,
            label={west:\contour*{white}{$c\;$}},
            label={east:\footnotesize\color{gray}$\displaystyle -\frac{\rho^2 N}{2}$}
        ] (c) at (3*\w+\wgap,0) {};

        \node[
            style=circ,
            label={north west:$x_1$},
            label={north east:\small\color{gray} $\displaystyle +\rho$}
        ] (v11) at (0,\h) {};

        \node[
            style=circ,
            label={north west:$x_2$},
            label={north east:\small\color{gray} $\displaystyle +\rho$}
        ] (v12) at (\w,\h) {};

        \node[
            style=circ,
            label={north west:$x_\rho$},
            label={north east:\small\color{gray} $\displaystyle +\rho$}
        ] (v1r) at (3*\w,\h) {};

        \begin{scope}[on background layer]
            \node[
                style=graph,
                minimum width=1,
                minimum height=2*\hgap cm
            ] at (0,\h+\hgap) {$G_1$};

            \node[
                style=graph,
                minimum width=1,
                minimum height=2*\hgap cm
            ] at (\w,\h+\hgap) {$G_2$};

            \node at (2*\w,\h+\hgap) {$\cdots$};

            \node[
                style=graph,
                minimum width=1,
                minimum height=2*\hgap cm
            ] at (3*\w,\h+\hgap) {$G_\rho$};
        \end{scope}

        \node[
            style=circ,
            label={south west:$y_1$},
            label={south east:\small\color{gray} $\displaystyle +\rho$}
        ] (v21) at (0,\h+2*\hgap) {};

        \node[
            style=circ,
            label={south west:$y_2$},
            label={south east:\small\color{gray} $\displaystyle +\rho$}
        ] (v22) at (\w,\h+2*\hgap) {};

        \node[
            style=circ,
            label={south west:$y_\rho$},
            label={south east:\small\color{gray} $\displaystyle +\rho$}
        ] (v2r) at (3*\w,\h+2*\hgap) {};

        \coordinate (l1) at (0,\h+2*\hgap+\hbot);
        \coordinate (r1) at (3*\w+\wgap,\h+2*\hgap+\hbot);

        \coordinate (l2) at (\w,\h+2*\hgap+\hbot-\epsdist);
        \coordinate (r2) at (3*\w+\wgap-\epsdist,\h+2*\hgap+\hbot-\epsdist);
        \coordinate (u2) at (3*\w+\wgap-\epsdist,1);

        \coordinate (lr) at (3*\w,\h+2*\hgap+\hbot-3*\epsdist);
        \coordinate (rr) at (3*\w+\wgap-3*\epsdist,\h+2*\hgap+\hbot-3*\epsdist);
        \coordinate (ur) at (3*\w+\wgap-3*\epsdist,1);

        \begin{scope}[on background layer]
            \draw[ultra thick] (a.center) -- (b1.center);
            \draw[ultra thick] (a.center) -- (b2.center);
            \draw[ultra thick] (a.center) -- (br.center);
            \draw[ultra thick] (a.center) -- (c.center);

            \draw (b1.center) -- (v11.center);
            \draw (b2.center) -- (v12.center);
            \draw (br.center) -- (v1r.center);

            \draw[rounded corners=8pt] (v21.center) -- (l1) -- (r1) -- (c.center);
            \draw[rounded corners=8pt] (v22.center) -- (l2) -- (r2) -- (u2) -- (c.center);
            \draw[rounded corners=8pt] (v2r.center) -- (lr) -- (rr) -- (ur) -- (c.center);
        \end{scope}

    \end{tikzpicture}

    \caption{%
        Illustration of \cref{cons:general-networks}.
        Demands are shown in gray.
        Thick edges have unit resistance, while thin edges have zero resistance.
    }
    \label{fig:general-network-reduction}
\end{figure}
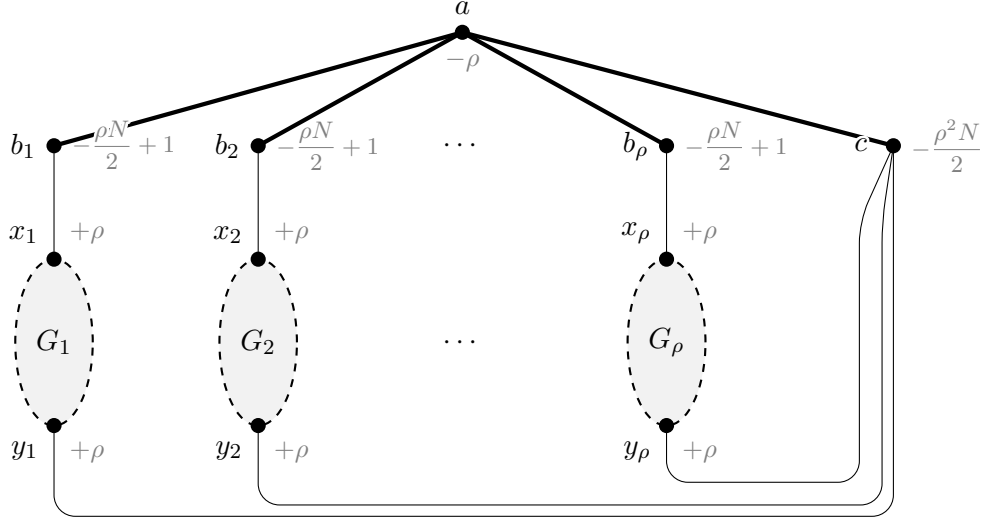

    Let $G'$ initially be empty, and let $N = |V(G)|$.
    If $\rho$ is not an integer, then replace it by~$\lceil \rho \rceil$.
    Create $\rho$ copies of $G$, named $G_1, \dots, G_\rho$, and add them to $G'$.
    We refer to the copy of vertex $x$ in graph $G_i$ as $x_i$ for each $i \in \{1, \dots, \rho\}$, and use analogous notation for the copies of $y$.
    Add the vertices $a, c$, and, for each $i \in \{1, \dots, \rho\}$, a vertex $b_i$.
    Add the edge sets $E_a \ceq \{ \{a,c\} \} \cup \bigcup_{i = 1}^\rho \{ \{a, b_i\} \}$, $E_B \ceq \bigcup_{i = 1}^\rho \{ \{b_i, x_i\} \}$, and $E_c \ceq \bigcup_{i = 1}^\rho \{ \{c, y_i\} \}$.

    Choose $d(a) \ceq -\rho$, $d(b_i) \ceq -\rho N/2 + 1$ for every $i \in \{1, \dots, \rho\}$, $d(c) \ceq -\rho^2 N / 2$, and $d(v) \ceq \rho$ for every $v \in \bigcup_{i=1}^{\rho} V(G_i)$.
    Finally, let $r(e) \ceq 1$ if $e \in E_a$, and $r(e) \ceq 0$ otherwise.
\end{construction}

\begin{proof}[Proof of \cref{thm:constant-factor-hardness}]
    It is straightforward to verify that $\sum_{v \in V(G')} d(v) = 0$.
    As evident from \cref{fig:general-network-reduction}, $G'$ is connected and planar.
    Moreover, for fixed $\rho$, all demands are polynomially bounded in the instance size.
    We establish the following gap: if $\III$ is a \ttup{yes}-instance, then $\OPT(\III') \leq \rho$, whereas if $\III$ is a \ttup{no}-instance, then $\OPT(\III') \geq \rho^2$.

    We first consider the case that $\III$ is a \ttup{yes}-instance.
    Then, there exists a solution $V_x,V_y$ to $\III$.
    We may assume that $x \in V_x$ and $y \in V_y$.
    Let $T_x$ be a spanning tree of $G[V_x]$ and let $T_y$ be a spanning tree of $G[V_y]$.
    For each $i \in \{1, \dots, \rho\}$, let $T_x^i$, $T_y^i$ denote the copies of~$T_x$ and $T_y$ in $G_i$.
    Now, let $E' \ceq E_a \cup E_B \cup E_c \cup \bigcup_{i = 1}^\rho ( E(T_x^i) \cup E(T_y^i) )$ and let $T'$ be the subgraph of $G'$ induced by $E'$.
    Note that $T'$ is a spanning tree of $G'$ and hence a solution to $\III'$.
    We have
    $
        f_{T'}(a,b_i) = d(b_i) + \sum_{v \in V(T_x^i)} d(v) = -\frac{\rho N}{2} + 1 + \frac{\rho N}{2} = 1
    $
    for each $i \in \{1, \dots, \rho\}$, and
    $
        f_{T'}(a,c) = d(c) + \sum_{i = 1}^\rho \sum_{v \in V(T_y^i)} d(v) = -\frac{\rho^2 N}{2} + \frac{\rho^2 N}{2} = 0
    $.
    Since only edges in $E_a$ have nonzero resistance, we conclude that $L(T') = \rho$.
    Therefore, $\OPT(\III') \leq \rho$.

    We now turn to the case that $\III$ is a \ttup{no}-instance.
    The source $a$ emits $\rho$ units of flow.
    If only one edge carries flow away from $a$, then this edge alone produces a loss of at least $\rho^2$.
    Hence, from now on, assume that at least two edges carry flow away from~$a$.
    Let $T'$ be a spanning tree of $G'$ and assume for contradiction that $L(T') < \rho^2$.
    Let $\{a,b_i\} \in E(T')$ be one of the edges carrying flow away from $a$.
    The flow on this edge must satisfy $1 \leq f_{T'}(a, b_i) < \rho$.

    Let $T_i'$ be the component of $T' - \{ a \}$ that contains $b_i$.
    We claim that $c \in V(T_i')$.
    Suppose not.
    Since both $a$ and $c$ are missing from $T_i'$, we have $V(T_i') = \{b_i\} \cup V_x$ for some $V_x \subseteq V(G_i)$.
    We find that
    \begin{displaymath}
        f_{T'}(a,b_i)
        = d(b_i) + \sum_{v \in V_x} d(v)
        = -\frac{\rho N}{2} + 1 + \rho |V_x|
        = 1 + \rho \left( |V_x| - \frac{N}{2} \right).
    \end{displaymath}
    The constraint $1 \leq f_{T'}(a, b_i) < \rho$ forces $|V_x| = N/2$.
    Let $V_y \ceq V(G_i) \setminus V_x$.
    Clearly, $x_i \in V_x$.
    Moreover, $y_i \in V_y$ holds because $V_y$ consists of sinks and $\{y_i,c\}$ is the only edge available to connect them to a source for flow conservation.
    Overall, note that $V_x, V_y$ directly corresponds to a solution to $\III$, a contradiction.
    This proves the claim that $c \in V(T_i')$.

    Now, consider a second edge carrying flow away from $a$.
    If it is $\{a,c\}$, then $\{a,b_i\}$, the $(b_i,c)$-path in $T_i'$, and $\{a,c\}$ form a cycle in $T'$.
    If it is $\{a,b_j\}$ for some $j \neq i$, then the claim places both $b_i$ and $b_j$ in the component of $T' - \{a\}$ containing $c$.
    Hence, $\{a,b_i\}$, the $(b_i,b_j)$-path, and $\{a,b_j\}$ form a cycle in $T'$.
    Either way, $T'$ is not a tree, a contradiction.
    We conclude that $L(T') \geq \rho^2$, and hence $\OPT(\III') \geq \rho^2$.
\end{proof}

The parameter $\rho$ does not need to be a constant: the construction can still be performed in polynomial time if $\rho$ is polynomially bounded in the size of the \scup{BCP} instance.
Using this observation, we obtain the following stronger result, where $n$ denotes the number of vertices of the \scup{DNR} instance.

\begin{corollary}\label{cor:general-hardness}
    For every $\eps > 0$, there is no $n^{1-\eps}$-approximation algorithm for \scup{DNR} on planar graphs unless $\mathrm{P} = \mathrm{NP}$.
    This holds even if each resistance is either $0$ or $1$ and all demands are polynomially bounded in the instance size.
\end{corollary}
\begin{proof}
    It suffices to consider $0 < \eps \leq 1$.
    Choose a constant $\delta > 1/\eps$.
    Note that $\delta > 1$ and $\delta \eps > 1$.
    Recall that $N$ denotes the number of vertices in the \scup{BCP} instance.
    Let $\rho \ceq N^{\delta - 1}$.
    We find that
    \begin{displaymath}
        n
        = \lceil \rho \rceil (N + 1) + 2
        < 2 \rho N
        < (2 N)^\delta,
    \end{displaymath}
    where the first inequality holds for sufficiently large $N$.
    Moreover, we have
    \begin{displaymath}
        \rho
        = N^{\delta - 1}
        > (2 N)^{\delta - \delta\eps}
        = \left( (2 N)^\delta \right)^{1 - \eps}
        > n^{1 - \eps},
    \end{displaymath}
    where the first inequality again holds for sufficiently large $N$.
    Hence, by the same arguments as in the proof of \cref{thm:constant-factor-hardness}, we obtain an asymptotic gap of $\rho > n^{1 - \eps}$.
\end{proof}

The constructed instance has $\lceil \rho \rceil + 2$ sources.
Thus, \cref{thm:constant-factor-hardness} gives constant-factor inapproximability for a constant number of sources, where the number of sources depends on the factor.
\cref{cor:general-hardness} yields superconstant inapproximability, but the required $\rho$, and hence the number of sources, grow with the instance size.
Next, we show that superconstant inapproximability already holds for two sources.

\subsubsection{Hardness of Two-Source Distribution Network Reconfiguration}\label{sub:two-hardness}

We prove the following theorem.

\begin{theorem}\label{thm:2dnr-hardness}
    There is a constant $c > 0$ such that \scup{2-DNR} cannot be approximated within a factor of $c \log^2 n$ unless $\mathrm{P} = \mathrm{NP}$.
    This holds even if each resistance is either $0$ or $1$ and all demands are polynomially bounded in the instance size.
\end{theorem}

We perform an approximation-preserving reduction from \scup{Set Cover}.
We remark that this is not an L-reduction: the hardness transfer is nonlinear and hence requires a more tailored approach.

\probDefOptimization{Set Cover}
{A ground set $U= \{ u_1, \dots, u_\nu \}$ and a collection of subsets $\mathcal{S} = \{S_1, \dots, S_\mu\}$ with $S_i \subseteq U$ for every $i \in \{1, \dots, \mu\}$.}
{A subcollection $\mathcal{S}' \subseteq \mathcal{S}$ that covers every element of the ground set, that is, $\bigcup_{S \in \mathcal{S}'} S = U$.}
{Minimize $|\mathcal{S}'|$.}

There is a constant $c > 0$ such that \scup{Set Cover} cannot be approximated within a factor of $c \log (\nu + \mu)$ unless $\mathrm{P} = \mathrm{NP}$ \cite[Chapter 8.8]{demaine2026computational}.
We restrict our attention to nontrivial instances of \scup{Set Cover}.
These are instances with $U \neq \emptyset$ that have a solution, that is, each element of $U$ appears in at least one subset of $\mathcal{S}$.

We first present the forward mapping that transforms any instance $\III$ of \scup{Set Cover} into an instance $\III'$ of \scup{2-DNR}.
In the following, we construct a graph containing $\mu$ copies, indexed by $\ell$, of the \scup{Set Cover}-instance. 
The two sources $a$ and $b$ do not belong to any particular copy.

\begin{construction}[Forward Mapping]\label{cons:2dnr-forward}
    Let $\III = (U, \mathcal{S})$ be an instance of \scup{Set Cover} with $U = \{u_1, \dots, u_\nu\}$ and $\mathcal{S} = \{S_1, \dots, S_\mu\}$.
    We construct an instance $\III' = (G=(V,E), d, r)$ of \scup{2-DNR} as follows (see \cref{fig:two-source-reduction} for an illustration).

    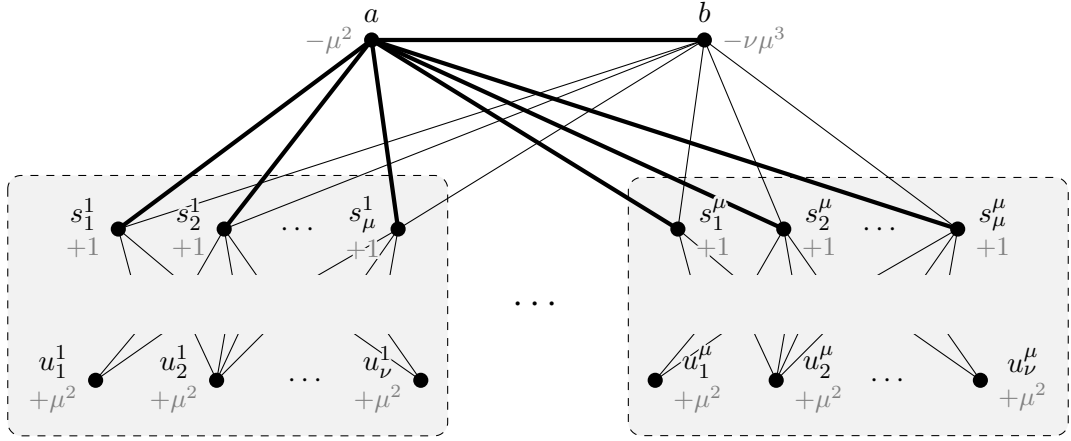
\begin{figure}[t]
    \centering

    \begin{tikzpicture}[
            yscale=-1,
            circ/.style={fill,circle,inner sep=0pt,minimum size=2mm}
        ]

        \tikzmath{
            \w = 1.4;
            \sgap = 1;
            \wgap = 1.3;
            \wtop = 2.2;
            \epsdist = 0.1;
            \h = 2;
            \htop = 2.5;
            \pad = 0.15;
            \parts = 7.4;
            \stub = 0.6;
        }
        \contourlength{0.3mm}

        \begin{scope}[local bounding box=lbox]

            \node[
                style=circ,
                label={west:\contour*{gray!10}{\shortstack{$s_1^1$\\\small\color{gray}$+1$}}}
            ] (s11) at (0,0) {};

            \node[
                style=circ,
                label={west:\contour*{gray!10}{\shortstack{$s_2^1$\\\small\color{gray}$+1$}}}
            ] (s21) at (\w,0) {};

            \node at (\w+\sgap,0) {$\cdots$};

            \node[
                style=circ,
                label={west:\contour*{gray!10}{\shortstack{$s_\mu^1$\\\small\color{gray}$+1$}}}
            ] (sm1) at (\w+\sgap+\wgap,0) {};

            \node[
                style=circ,
                label={west:\contour*{gray!10}{\shortstack{$u_1^1$\\\small\color{gray}$+\mu^2$}}}
            ] (u11) at (0-3*\epsdist,\h) {};

            \node[
                style=circ,
                label={west:\contour*{gray!10}{\shortstack{$u_2^1$\\\small\color{gray}$+\mu^2$}}}
            ] (u21) at (\w-\epsdist,\h) {};

            \node at (\w+\sgap+\epsdist,\h) {$\cdots$};

            \node[
                style=circ,
                label={west:\contour*{gray!10}{\shortstack{$u_\nu^1$\\\small\color{gray}$+\mu^2$}}}
            ] (un1) at (\w+\sgap+\wgap+3*\epsdist,\h) {};

        \end{scope}

        \begin{scope}[xshift=\parts cm,local bounding box=rbox]

            \node[
                style=circ,
                label={east:\contour*{gray!10}{\shortstack{$s_1^\mu$\\\small\color{gray}$+1$}}}
            ] (s1m) at (0,0) {};

            \node[
                style=circ,
                label={east:\contour*{gray!10}{\shortstack{$s_2^\mu$\\\small\color{gray}$+1$}}}
            ] (s2m) at (\w,0) {};

            \node at (\w+\wgap,0) {$\cdots$};

            \node[
                style=circ,
                label={east:\contour*{gray!10}{\shortstack{$s_\mu^\mu$\\\small\color{gray}$+1$}}}
            ] (smm) at (\w+\sgap+\wgap,0) {};

            \node[
                style=circ,
                label={east:\contour*{gray!10}{\shortstack{$u_1^\mu$\\\small\color{gray}$+\mu^2$}}}
            ] (u1m) at (0-3*\epsdist,\h) {};

            \node[
                style=circ,
                label={east:\contour*{gray!10}{\shortstack{$u_2^\mu$\\\small\color{gray}$+\mu^2$}}}
            ] (u2m) at (\w-\epsdist,\h) {};

            \node at (\w+\wgap+\epsdist,\h) {$\cdots$};

            \node[
                style=circ,
                label={east:\contour*{gray!10}{\shortstack{$u_\nu^\mu$\\\small\color{gray}$+\mu^2$}}}
            ] (unm) at (\w+\sgap+\wgap+3*\epsdist,\h) {};

        \end{scope}

        \begin{scope}[on background layer]
            \draw[dashed,rounded corners=5pt,fill=gray!10] ([xshift=-\pad cm,yshift=-\pad cm]lbox.north west) rectangle ([xshift=3+\pad cm,yshift=\pad cm]lbox.south east);

        \draw[dashed,rounded corners=5pt,fill=gray!10] ([xshift=-3-\pad cm,yshift=-\pad cm]rbox.north west) rectangle ([xshift=\pad cm,yshift=\pad cm]rbox.south east);
        \end{scope}

        \coordinate (mid) at ($(sm1)!0.5!(s1m)$);

        \node[
            style=circ,
            label={north:$a$},
            label={west:\small\color{gray}$-\mu^2$}
        ] (a) at ([xshift=-\wtop cm,yshift=-\htop cm]mid) {};

        \node[
            style=circ,
            label={north:$b$},
            label={east:\small\color{gray}$-\nu\mu^3$}
        ] (b) at ([xshift=\wtop cm,yshift=-\htop cm]mid) {};

        \node at ([yshift=0.5*\h cm]mid) {\Large $\cdots$};

        \begin{scope}[on background layer]
            \draw[ultra thick] (a.center) -- (s11.center);
            \draw[ultra thick] (a.center) -- (s21.center);
            \draw[ultra thick] (a.center) -- (sm1.center);

            \draw[ultra thick] (a.center) -- (s1m.center);
            \draw[ultra thick] (a.center) -- (s2m.center);
            \draw[ultra thick] (a.center) -- (smm.center);

            \draw[ultra thick] (a.center) -- (b.center);

            \draw (b.center) -- (s11.center);
            \draw (b.center) -- (s21.center);
            \draw (b.center) -- (sm1.center);

            \draw (b.center) -- (s1m.center);
            \draw (b.center) -- (s2m.center);
            \draw (b.center) -- (smm.center);

            \begin{scope}[overlay]
                \clip (-\w,-\stub) rectangle (4*\w,\stub);

                \draw (s11.center) -- ++(75:2);
                \draw (s11.center) -- ++(40:2);

                \draw (s21.center) -- ++(120:2);
                \draw (s21.center) -- ++(80:2);
                \draw (s21.center) -- ++(55:2);

                \draw (sm1.center) -- ++(150:2);
                \draw (sm1.center) -- ++(125:2);
                \draw (sm1.center) -- ++(100:2);
            \end{scope}

            \begin{scope}[overlay]
                \clip (-\w,-\stub+\h) rectangle (4*\w,\stub+\h);

                \draw (u11.center) -- ++(-35:2);
                \draw (u11.center) -- ++(-50:2);

                \draw (u21.center) -- ++(-45:2);
                \draw (u21.center) -- ++(-65:2);
                \draw (u21.center) -- ++(-80:2);
                \draw (u21.center) -- ++(-115:2);

                \draw (un1.center) -- ++(-145:2);
                \draw (un1.center) -- ++(-125:2);
            \end{scope}

            \begin{scope}[xshift=\parts cm,overlay]
                \clip (-\w,-\stub) rectangle (4*\w,\stub);

                \draw (s1m.center) -- ++(75:2);
                \draw (s1m.center) -- ++(40:2);

                \draw (s2m.center) -- ++(120:2);
                \draw (s2m.center) -- ++(80:2);
                \draw (s2m.center) -- ++(55:2);

                \draw (smm.center) -- ++(150:2);
                \draw (smm.center) -- ++(125:2);
                \draw (smm.center) -- ++(100:2);
            \end{scope}

            \begin{scope}[xshift=\parts cm,overlay]
                \clip (-\w,-\stub+\h) rectangle (4*\w,\stub+\h);

                \draw (u1m.center) -- ++(-35:2);
                \draw (u1m.center) -- ++(-50:2);

                \draw (u2m.center) -- ++(-45:2);
                \draw (u2m.center) -- ++(-65:2);
                \draw (u2m.center) -- ++(-80:2);
                \draw (u2m.center) -- ++(-115:2);

                \draw (unm.center) -- ++(-145:2);
                \draw (unm.center) -- ++(-125:2);
            \end{scope}
        \end{scope}

    \end{tikzpicture}

    \caption{%
        Illustration of \cref{cons:2dnr-forward}.
        Demands are shown in gray.
        Thick edges have unit resistance, while thin edges have zero resistance.
        Each framed part represents one copy.
    }
    \label{fig:two-source-reduction}
\end{figure}

    Let $V_\mathcal{S}^\ell \ceq \bigcup_{i=1}^\mu \{s_i^\ell\}$ for each $\ell \in \{1, \dots, \mu\}$, and let $V_\mathcal{S} \ceq \bigcup_{\ell=1}^\mu V_\mathcal{S}^\ell$.
    Similarly, let $V_U^\ell \ceq \bigcup_{j=1}^\nu \{u_j^\ell\}$ for each $\ell \in \{1, \dots, \mu\}$, and let $V_U \ceq \bigcup_{\ell=1}^\mu V_U^\ell$.
    Let $V \ceq \{a,b\} \cup V_\mathcal{S} \cup V_U$ be the full vertex set.
    Let $E_a \ceq \{ \{a, v\} \mid v \in V_\mathcal{S} \}$ and $E_b \ceq \{ \{b, v\} \mid v \in V_\mathcal{S} \}$.
    The edge set $E_\SU \ceq \bigcup_{\ell=1}^\mu \{ \{ s_i^\ell, u_j^\ell \} \mid u_j \in S_i \}$ encodes membership of the elements of the ground set~$U$ within the subsets in the collection $\mathcal{S}$.
    Let $E \ceq \{\{a,b\}\} \cup E_a \cup E_b \cup E_\SU$ be the full edge set.

    Let $a$ and $b$ be sources with $d(a) \ceq - \mu^2$ and $d(b) \ceq - \nu \mu^3$.
    Let all other vertices be sinks with $d(v) \ceq 1$ for every $v \in V_\mathcal{S}$, and $d(v) \ceq \mu^2$ for every $v \in V_U$.
    Finally, let $r(e) \ceq 1$ if $e \in E_a \cup \{ \{a,b\} \}$, and $r(e) \ceq 0$ otherwise.
\end{construction}

We have $|V_\mathcal{S}| = \mu^2$, $|V_U| = \nu \mu$, and $|V| = \mu^2 + \nu \mu + 2$.
Note that $G$ is connected and that $\sum_{v \in V} d(v) = 0$.
Moreover, all demands are polynomially bounded in the instance size.

The idea behind the construction is the following.
Only the edges incident to source $a$ have nonzero resistance.
To minimize loss, we would like to distribute the flow originating at $a$ to its neighbors, the vertex set $V_\mathcal{S} \cup \{b\}$, as evenly as possible.
However, source $b$ emits a much larger quantity of flow, most of which must be routed to the high-demand sinks in~$V_U$.
To reach the vertices in $V_U$, the flow originating at $b$ must pass through some vertices in~$V_\mathcal{S}$.
Because of the required acyclicity, these pass-through vertices in $V_\mathcal{S}$ cannot, with the exception of at most one, be adjacent to $a$ in a spanning tree solution.
Therefore, minimizing loss is directly related to minimizing the number of pass-through vertices, which need to ``cover'' the vertices in~$V_U$.
This establishes, on a high level, the correspondence to the \scup{Set Cover} problem.

We remark that the purpose of having multiple copies of the same structure is to amplify the loss penalty incurred by pass-through vertices.
This sharpens the loss gap between solutions of different quality, especially when the number of pass-through vertices is small.

Next, we present the backward mapping.
Given a tree $T$, let $\Int(T) \subseteq V(T)$ denote the internal (non-leaf) vertices of $T$.
Recall that every solution $x'$ to $\III'$ is a tree.
Intuitively, for a specific solution $x'$, the set $ V_\mathcal{S} \cap \Int(x')$ contains the pass-through vertices.
The reason is that, in low-cost solutions, vertices in $V_\mathcal{S}$ that have no pass-through role are leaves attached to source $a$.

\begin{construction}[Backward Mapping]\label{map:2dnr-backward}
    Let $x'$ be a solution to $\III'$.
    We construct a solution $x$ to $\III$.
    Choose an index $\ell \in \{1, \dots, \mu\}$ such that the size of $ V_\mathcal{S}^\ell \cap \Int(x')$ is minimal.
    Then, put the elements of $\mathcal{S}$ corresponding to the vertices in $ V_\mathcal{S}^\ell \cap \Int(x')$ into the solution $x$, that is, let $x \ceq \{S_i \subseteq \mathcal{S} \mid s_i^\ell \in  V_\mathcal{S}^\ell \cap \Int(x')\}$.
\end{construction}

We now show that $x$ is a solution to $\III$.
Assume for contradiction that $x$ is not a solution to $\III$.
Then, some element $u_j \in U$ is not covered by $x$, that is, $u_j \notin S_i$ for every $S_i \in x$.
Recall that $E_\SU = \bigcup_{\ell=1}^\mu \{ \{ s_i^\ell, u_j^\ell \} \mid u_j \in S_i \}$.
Let $\ell \in \{1,\dots,\mu\}$ be the index chosen by the backward mapping.
For each vertex $s_i^\ell \in V_\mathcal{S}^\ell$ with $\{s_i^\ell,u_j^\ell\} \in E_\SU$, we have $u_j \in S_i$ by the definition of $E_\SU$, which implies $S_i \notin x$ by the definition of $u_j$, and finally implies $s_i^\ell \notin \Int(x')$ by the definition of $x$.
Every path in the tree $x'$ starting at $u_j^\ell$ begins with some edge $\{u_j^\ell, s_i^\ell\} \in E_\SU$, and must end at $s_i^\ell$ since $s_i^\ell \notin \Int(x')$.
In particular, there is no path from $u_j^\ell$ to any of the sources $a$ and $b$ in $x'$, implying that $x'$ is not a spanning tree of~$G$, a contradiction.

The following lemma shows that pass-through vertices contribute to loss.
Specifically, the square of the number of pass-through vertices is a lower bound on the loss.

\begin{lemma}\label{lem:2dnr-first}
    Let $x'$ be a solution to $\III'$.
    Then, $L(x') \geq |  V_\mathcal{S} \cap \Int(x') |^2$.
\end{lemma}
\begin{proof}
    Let $T_1, \dots, T_k$ be the components of $x' - \{a\}$, and let $\mathcal{C} \ceq \{T_1, \dots, T_k\}$.
    Since exactly the edges incident to $a$ have nonzero resistance, we have
    \begin{displaymath}
        L(x')
        = \sum_{\{a,v\} \in E(x')} \big( f_{x'}(a,v) \big)^2
        = \sum_{T \in \mathcal{C}} \left( \sum_{v \in V(T)} d(v) \right)^2.
    \end{displaymath}
    Let $T_b \in \mathcal{C}$ be the component containing $b$.
    We distinguish between two cases.

    First, consider the case that every $v \in V_U$ is contained in $T_b$.
    By construction, each vertex in $V_\mathcal{S}$ has the following neighbors in the full graph $G$: the sources $a,b$, and some vertices from $V_U$.
    Each vertex in $\Int(x')$ has a degree of at least two in $x'$, and hence a degree of at least one in $x' - \{a\}$.
    We conclude that each vertex in $V_\mathcal{S} \cap \Int(x')$ is adjacent to $b$ or some vertex of $V_U$ in $x' - \{a\}$.
    Since $b$ and every vertex in $V_U$ are contained in $T_b$, this implies that all vertices in $V_\mathcal{S} \cap \Int(x')$ are also contained in $T_b$.
    In summary, we have $\{b\} \cup \big( V_\mathcal{S} \cap \Int(x') \big) \cup V_U \subseteq V(T_b)$.
    Using this, we get
    \begin{align*}
        \sum_{v \in V(T_b)} d(v)
        &\geq d(b) + | V_\mathcal{S} \cap \Int(x')| \cdot 1 + |V_U| \cdot \mu^2 \\
        &= -\nu \mu^3 + | V_\mathcal{S} \cap \Int(x') | + \nu \mu^3 \\
        &= | V_\mathcal{S} \cap \Int(x') |,
    \end{align*}
    where the inequality holds because all vertices of $T_b$ other than $b$ have positive demand.
    We find that
    \begin{displaymath}
        L(x')
        = \sum_{T \in \mathcal{C}} \left( \sum_{v \in V(T)} d(v) \right)^2
        \geq \left( \sum_{v \in V(T_b)} d(v) \right)^2
        \geq | V_\mathcal{S} \cap \Int(x') |^2.
    \end{displaymath}

    Now, consider the case that some $v \in V_U$ is not contained in $T_b$.
    In other words, there is a component $\widetilde{T} \in \mathcal{C} \setminus \{ T_b \}$ and a vertex $\tilde{v} \in V_U$ such that $\tilde{v} \in V(\widetilde{T})$.
    We have
    \begin{displaymath}
        \sum_{v \in V(\widetilde{T})} d(v)
        = d(\tilde{v}) + \sum_{\substack{v \in V(\widetilde{T}) \\ v \notin \{a,b,\tilde{v}\}}} d(v)
        \geq d(\tilde{v})
        = \mu^2.
    \end{displaymath}
    Finally,
    \begin{displaymath}
        L(x')
        = \sum_{T \in \mathcal{C}} \left( \sum_{v \in V(T)} d(v) \right)^2
        \geq \left( \sum_{v \in V(\widetilde{T})} d(v) \right)^2
        \geq ( \mu^2 )^2
        = |V_\mathcal{S}|^2
        \geq | V_\mathcal{S} \cap \Int(x') |^2.\qedhere
    \end{displaymath}
\end{proof}

The next lemma bounds the optimal loss of the constructed instance $\III'$ in terms of the optimal size of a set cover for $\III$.

\begin{lemma}\label{lem:2dnr-second}
    The optimal values of $\III$ and $\III'$ satisfy $\OPT(\III') \leq 2\mu^2 \OPT(\III)^2$.
\end{lemma}
\begin{proof}
    Let $x$ be an optimal solution to $\III$.
    We prove the claim by constructing a solution $x'$ to $\III'$ such that $L(x') \le 2\mu^2 |x|^2$.

    Let $E_a' \ceq \bigcup_{\ell = 1}^\mu \{ \{ a, s_i^\ell \} \mid S_i \notin x \}$ and $E_b' \ceq \bigcup_{\ell = 1}^\mu \{ \{ b, s_i^\ell \} \mid S_i \in x \}$.
    Let $g \colon U \rightarrow x$ be an auxiliary function that maps each element $u \in U$ to a subset $S \in x$ such that $u \in S$.
    Such a function exists since $x$ is a solution to $\III$.
    Let $E_\SU' \ceq \bigcup_{\ell = 1}^\mu \{ \{ s_i^\ell, u_j^\ell \} \mid S_i = g(u_j) \}$.
    Finally, let $E' \ceq \{\{a,b\}\} \cup E_a' \cup E_b' \cup E_\SU'$ and let $x'$ be the subgraph of $G$ induced by $E'$.

    We argue that $x'$ is a spanning tree.
    First, $x'$ is connected: vertex $a$ is adjacent to vertex~$b$, each vertex in $V_\mathcal{S}$ is adjacent to $a$ or $b$, and each vertex in $V_U$ is adjacent to a vertex in~$V_\mathcal{S}$.
    Second, $x'$ has $|V| - 1$ edges: $|E'| = 1 + |E_a'| + |E_b'| + |E_\SU'| = 1 + |V_\mathcal{S}| + |V_U| = |V| - 1$.
    Hence, $x'$ is a solution to $\III'$.

    Next, we show that $L(x') \leq 2 \mu^2 |x|^2$.
    It suffices to consider the edges in $E_a'$ and the edge $\{a,b\}$ since only these edges of $x'$ have nonzero resistance.
    We begin with $E_a'$.
    Let $\{ a, s_i^\ell \} \in E_a'$.
    By definition of $E_a'$, we have $S_i \notin x$.
    Thus, applying the definitions of $E_b'$ and $E_\SU'$, we find that $s_i^\ell$ is adjacent only to $a$.
    Therefore, $f_{x'}(a,s_i^\ell) = d(s_i^\ell) = 1$ holds for every edge $\{ a, s_i^\ell \} \in E_a'$.
    Moreover, by flow conservation we have
    \begin{displaymath}
        f_{x'}(a,b)
        = -\left( d(a) + \sum_{s_i^\ell \colon \{a,s_i^\ell\} \in E_a'} d(s_i^\ell) \right)
        = \mu^2 - |E_a'|.
    \end{displaymath}
    We find that
    \begin{align*}
        L(x')
        &= \big( f_{x'}(a,b) \big)^2 + \sum_{ \{a,s_i^\ell\} \in E_a' } \big( f_{x'}(a,s_i^\ell) \big)^2 \\
        &= \big( \mu^2 - |E_a'| \big)^2 + |E_a'| \\
        &= \big( \mu^2 - \mu^2 + \mu |x| \big)^2 + \mu^2 - \mu |x| \\
        &= \mu^2 |x|^2 + \mu^2 - \mu |x| \\
        &\leq 2 \mu^2 |x|^2,
    \end{align*}
    where the third equality uses $|E_a'| = \mu (\mu - |x|) = \mu^2 - \mu |x|$ and the inequality uses that $|x| \geq 1$ for nontrivial instances of \scup{Set Cover}.

    Finally, note that $\OPT(\III') \leq L(x') \leq 2 \mu^2 |x|^2 = 2 \mu^2 \OPT(\III)^2$.
\end{proof}

We prove the key lemma next.
It states that if $x'$ is a low-loss solution, then $x$ is of small size.

\begin{lemma}\label{lem:2dnr-final}
    Let $x'$ be a solution to $\III'$.
    Let $x$ be a solution to $\III$ obtained from $x'$ by the backward mapping.
    If $L(x') \leq \rho \OPT(\III')$ for some $\rho \geq 1$, then $|x| \leq \sqrt{2 \rho} \OPT(\III)$.
\end{lemma}
\begin{proof}
    \begin{align*}
        |x|
        &= \min_{\ell=1}^\mu | V_\mathcal{S}^\ell \cap \Int(x')| && \text{(\cref{map:2dnr-backward})} \\
        &\leq \frac{1}{\mu} | V_\mathcal{S} \cap \Int(x')| \\
        &\leq \frac{1}{\mu} \sqrt{L(x')} && \text{(\cref{lem:2dnr-first})} \\
        &\leq \frac{1}{\mu} \sqrt{\rho \OPT(\III')} \\
        &\leq \frac{1}{\mu} \sqrt{2 \rho \mu^2 \OPT(\III)^2} && \text{(\cref{lem:2dnr-second})} \\
        &= \sqrt{2 \rho} \OPT(\III) && \qedhere
    \end{align*}
\end{proof}

Notably, the gap to the optimum tightens: if $x'$ deviates from the optimum by a factor of at most $\rho$, then $x$ deviates from the optimum by a factor of at most $\sqrt{2\rho}$.
We use this to prove the theorem.

\begin{proof}[Proof of \cref{thm:2dnr-hardness}]
    Assume for contradiction that a ($c \log^2 n$)-approximation algorithm for \scup{2-DNR} exists for every $c > 0$.
    Consider some $c > 0$.
    Given an instance $\III$ of \scup{Set Cover}, we do the following.
    First, we use the polynomial-time forward mapping to produce an instance $\III'$ of \scup{2-DNR}.
    Then, we use the ($c \log^2 n$)-approximation algorithm to compute a solution $x'$ to $\III'$ with $L(x') \leq (c \log^2 n) \cdot \OPT(\III')$ in polynomial time.
    Finally, we run the polynomial-time backward mapping to transform $x'$ into a solution $x$ to $\III$.
    Note that executing the entire sequence of steps only takes polynomial time.

    By \cref{lem:2dnr-final}, we have
    \begin{align*}
        |x|
        &\leq \sqrt{2c \log^2 n} \cdot \OPT(\III) \\
        &= \sqrt{2c} \cdot \log(n) \cdot \OPT(\III) \\
        &= \sqrt{2c} \cdot \log(\mu^2 + \nu \mu + 2) \cdot \OPT(\III) \\
        &< 3 \sqrt{2c} \cdot \log(\nu + \mu) \cdot \OPT(\III),
    \end{align*}
    where the last inequality uses $\nu \geq 1$ and $\mu \geq 1$.
    By appropriately choosing $c > 0$, the term $3 \sqrt{2c}$ can be made arbitrarily close to zero.
    This implies that \scup{Set Cover} can be approximated within a factor of $\tilde{c} \log(\nu + \mu)$ for any $\tilde{c} > 0$, a contradiction unless $\mathrm{P} = \mathrm{NP}$.
\end{proof}

\section{Outlook}

We extended the approximability study of \scup{DNR} in two directions.
For the single-source setting and uniform resistances, we improved the best known approximation bound from $\bigO(n)$ to $\bigO(\sqrt{n})$.
The large gap to the lower bound in this setting underlines the need for further research.
Moreover, we provided the first approximability for the multi-source settings.
Here, the approximation gap is closed, with the lower bound holding even for planar graphs.
Further analyzing the bounds of tractability of \scup{DNR} with respect to structural graph parameters such as the treewidth or feedback edge number
would nicely complement the theoretical contributions to \scup{DNR},
not least because the number of switchable edges beyond a spanning three is small in practice.
Further, to get closer to the real-world scenario, one could add operational constraints such as thermal line limits (edge capacities) or lower bounds on the voltage of the vertices (derivable from the flow).
Finally, one may ask not for a loss-minimal configuration but for a sequence of single switching actions transforming a given configuration into a target one, with every intermediate configuration feasible (e.g. line and voltage limits satisfied).

\section*{Acknowledgments}

We thank Aikaterini Niklanovits and Stefan Schmid for helpful discussions and comments.
This paper was written as part of the project ``Quantum-based Energy Grids (QuGrids)'', which receives funding from the program ``Profilbildung 2022'', an initiative of the Ministry of Culture and Science of the State of North Rhine-Westphalia.

\sloppy
\printbibliography

\end{document}